\documentclass[11pt]{article}
\usepackage{graphicx}
\usepackage{graphics,epsfig}
\usepackage[latin1]{inputenc}
\usepackage[english,activeacute]{babel}
\usepackage{latexsym}
\usepackage{amsmath, amsthm, amsfonts, amssymb}
\usepackage{verbatim}

\newcommand{\nd}{\noindent}
\newtheorem{theorem}{Theorem}

\begin{document}

\newtheorem{theo}{Theorem}[section]
\newtheorem{definition}[theo]{Definition}
\newtheorem{lem}[theo]{Lemma}
\newtheorem{prop}[theo]{Proposition}
\newtheorem{coro}[theo]{Corollary}
\newtheorem{exam}[theo]{Example}
\newtheorem{rema}[theo]{Remark}
\newtheorem{example}[theo]{Example}
\newtheorem{principle}[theo]{Principle}
\newcommand{\ninv}{\mathord{\sim}}
\newtheorem{axiom}[theo]{Axiom}
\title{A Quantum Version of Spectral Decomposition Theorem of Dynamical Systems, Quantum Chaos Hierarchy: Ergodic, Mixing and Exact}

\author{{\sc Ignacio Gomez}$^1$ {\sc and} \ {\sc Mario Castagnino}$^2$}

\maketitle

\begin{small}
\begin{center}
1- Instituto de F\'{\i}sica de Rosario (IFIR-CONICET), Rosario, Argentina\\
2- Instituto de Física de Rosario (IFIR-CONICET) and \\[0pt]
Instituto de Astronomía y Física del Espacio (IAFE-CONICET), \\[0pt]
Casilla de Correos 67, Sucursal 28, 1428 Buenos Aires, Argentina.\\[0pt]
\end{center}
\end{small}

\vspace{1cm}

\begin{abstract}
In this paper we study Spectral Decomposition Theorem \cite{LM}
and translate it to quantum language by means of the Wigner
transform. We obtain a quantum version of Spectral Decomposition
Theorem (QSDT) which enables us to achieve three distinct goals: First, to rank Quantum Ergodic Hierarchy levels \cite{0, NACHOSKY MARIO}. Second, to analyze
the classical limit in quantum ergodic systems and quantum mixing
systems. And third, and maybe most important feature, to find a
relevant and simple connection between the first three levels of
quantum ergodic hierarchy (ergodic, exact and mixing) and quantum spectrum. Finally, we illustrate the physical relevance of QSDT applying it to two examples: Microwave billiards \cite{stockmann, ste95} and a phenomenological Gamow model type \cite{Omnes1, Omnes2} .
\end{abstract}

\begin{small}
\centerline{\em Key words: QSDT-QEH-ergodic-mixing-exact}
\end{small}
\bigskip
\noindent

\bibliography{pom}

\begin{thebibliography}{10}

\bibitem{LM} A. Lasota, M. Mackey, \textit{Probabilistic properties of
deterministic systems, }Cambridge Univ. Press, Cambridge, 1985.

\bibitem{0} M. Castagnino, O. Lombardi , \textit{Phys. A,} \textbf{388}, 247-267, 2009.

\bibitem{NACHOSKY MARIO} I. Gomez, M. Castagnino, \emph{Towards a definition of the Quantum Ergodic Hierarchy: Kolmogorov and Bernoulli systems}, \textit{Physica A}, \textbf{393}, 112-131, 2014.

\bibitem{stockmann} H. Stockmann, \textit{Quantum Chaos - An Introduction}, Cambridge Univ. Press, Cambridge, 1999.

\bibitem{ste95} U. Stoffregen, J. Stein, H. Stockmann, M. Kus, F. Haake, \emph{Phys. Rev. Lett.}, \textbf{74}, 2666, 1995.

\bibitem {Omnes1} R. Laura, M. Castagnino, \textit{Phys. Rev. E}, \textbf{57},
3948-3961, 1998.

\bibitem{Omnes2} R. Omn\`{e}s, \textit{The Interpretation of Quantum Mechanics}, Princeton University, Princeton, 1994.

\bibitem{BELOT} G. Belot, J. Earman, \emph{Chaos out order: Quantum
mechanics, the correspondence principle, and chaos}, \emph{Stud. His.
Philos. Mod. Phys.} \textbf{28}, 147-182, 1997.

\bibitem{lich} A. Lichtenberg, M. Lieberman, \textit{Regular and Chaotic Dynamics}, Springer, New York, 1992.

\bibitem{guck} J. Guckenheimer, P. Holmes, \textit{Nonlinear oscillations, dynamical systems, and bifurcations of vector fields}, Springer, Ithaca, 1985.

\bibitem{3M} J. Berkovitz, R. Frigg, F. Kronz, \textit{Stud. Hist. Phil.
Mod. Phys.,} \textbf{37}, 661-691, 2006.

\bibitem{gutzwiller} M. Gutzwiller, \textit{Chaos in Classical and Quantum Mechanics}, Springer Verlag, New York, 1990.

\bibitem{kuramoto1} Y. Kuramoto, H. Araki \textit{Lecture Notes in Physics, International Symposium on Mathematical Problems in Theoretical Physics 39}, Springer Verlag, New York, 1975.

\bibitem{kuramoto2} Y. Kuramoto, \textit{Chemical Oscillations, Waves, and Turbulence}, Springer Verlag, New York, 1984.

\bibitem{haake} F. Haake, \textit{Quantum Signatures of Chaos,} 2nd edition, Springer-Verlag, Heidelberg, 2001.

\bibitem{casati} G. Casati, B. Chirikov, \textit{Quantum Chaos: between order and disorder,} Cambridge Univ. Press, Cambridge, 1995.

\bibitem{tabor} M. Tabor, \textit{Chaos and Integrability in Nonlinear Dynamics: An Introduction}, Wiley, New York, 1988.

\bibitem{BERRY} M. Berry, \textit{Physica Scripta}, \textbf{40}, 335-336, 1989.

\bibitem{MANTICA1} J. Ford, G. Mantica, G. Ristow, \emph{The Arnol'd cat:
Failure of the correspondence principle}, \emph{Physica D}, \textbf{50}, 493-520, 1991.

\bibitem{MANTICA2} J. Ford, G. Mantica, \emph{Does quantum mechanics obey the correspondence principle? Is it complete?}, \emph{Amer. J. Phys.}, \textbf{60}, 1086-1098, 1992.

\bibitem{SCHUSTER} H. Schuster, \emph{Deterministic Chaos}, VCH, Weinheim,
1984.

\bibitem{BATTERMAN} R. Batterman, \emph{Chaos, quantization and the
correspondence principle}, \emph{Synthese}, \textbf{89}, 189-227, 1991.

\bibitem{WEINBERG} S. Weinberg, \emph{Testing quantum mechanics},
\emph{Ann. Phys.}, \textbf{194}, 335-336, 1989.

\bibitem{GHIRARDI} G. Ghirardi, A. Rimini, T. Weber, \emph{Unified dynamics
for microscopic and macroscopic systems}, \emph{Phys. Rev. D}, \textbf{34}, 470-491, 1986.

\bibitem{FOP} I. Gomez, M. Castagnino, \textit{On the classical limit of quantum
mechanics, fundamental graininess and chaos: Compatibility of chaos with the
correspondence principle}, \textit{Chaos, Solitons and Fractals}, \textbf{68}, 98-113, 2014.

\bibitem{MARIO OLIMPIA} M. Castagnino, O. Lombardi, \emph{Self-induced decoherence
and the classical limit of quantum mechanics}, Philos. Sci., \textbf{72}, 764-776, 2005.

\bibitem{sieber} M. Sieber, \textit{Pramana-Journal Of Physics}, \textbf{106}, 121-167, 1984.

\bibitem{kuhl} U. Kuhl, O. Legrand, F. Mortessagne, \emph{Progress of Physics}, \textbf{73}, 543-551, 2009.

\bibitem{rotten} I. Rotten, arXiv:0711.2926, 2007.

\bibitem{moiseyev} N. Moiseyev, \textit{Non-Hermitian Quantum Mechanics}, Cambridge Univ. Press, Cambridge, 2011.

\bibitem{zycz} W. Slomczynski, K. Zyczkowski, \textit{J. Math. Phys.}, \textbf{35}, 5674-5700, 1994.

\bibitem{zurek-paz} J. Paz, W. Zurek, \emph{Environment-induced decoherence and the transition from quantum to classical}, \emph{Course 8 of Les Houches Lectures Session LXXII: Coherent Atomic Matter Waves}, 533-614, Springer, Berlin, 2001.

\bibitem{HILLERY} M. Hillery, M. O'Connell, R. Scully, E. Wigner, \textit{Phys.
Rep}., \textbf{106}, 121-167, 1984.

\bibitem{DITO} G. Dito, D. Sternheimer, arxiv math. QA/0201168, 2002.

\bibitem{GADELLA} M. Gadella, \emph{Fortschr. Phys.}, \textbf{43}, 229-264, 1995.

\bibitem{case} W. Case, \emph{Am. J. Phys.}, \textbf{76}, 937, 2008.

\bibitem{antoniou} I. Antoniou, Z. Sucha- necki, R. Laura, S.
Tasaki, \emph{Intrinsic irreversibilty of quantum systems with diagonal singularity}, \textit{Physica A}, \textbf{241}, 737-772, 1997.

\bibitem{landau} L. Landau, E. Lifshitz, \emph{Quantum Mechanics: Non-Relativistic Theory}, Pergamon Press, England, 1977.

\bibitem{van hove1} L. van Hove, \textit{Physica A}, \textbf{20}, 603, 1954.

\bibitem{van hove2} L. van Hove, \textit{Physica A}, \textbf{25}, 268, 1959.

\bibitem{MARIO OLIMPIA2} M. Castagnino, O. Lombardi, \emph{Stud. Hist. Phil. Mod. Phys.}, 35, 73, 2004

\bibitem{Mah69} C. Mahaux, H. Weidenm\"{u}ller, \emph{Shell-Model Approach to Nuclear Reactions}, North-Holland, Amsterdan, 1969.

\bibitem{Lew91} C. Lewenkopf, H. Weidenm\"{u}ller, \emph{Ann. Phys.}, \textbf{212}, 53, 1991.

\bibitem{Fyo97} Y. Fyodorov, H. Sommers, \emph{J. Math. Phys.}, \textbf{38}, 1918, 1997.

\bibitem{Per96} E. Persson, T. Gorin, I. Rotter, \emph{Phys. Rev. E}, \textbf{59}, 3339, 1996.

\bibitem{Sok88} V. Sokoloff, V. Zelevinsky, \emph{Phys. Rev. Lett. B}, \textbf{202}, 10, 1988.

\bibitem{Sto98} H. Stockmann, P. Seba, \emph{J. Phys. A}, \textbf{31}, 3439, 1998.

\bibitem {Gadella} M. Gadella, G. Pronko, \textit{Fortschritte der Physik},
\textbf{59}, 795-859, 2011.

\bibitem {letterpolos} M. Castagnino, S. Fortin, \textit{Modern Physics
Letters A}, \textbf{26}, 2365-2373, 2011.

\bibitem {Ordonito-dec} G. Ordo\~{n}ez, S. Kim, \textit{Phys. Rev.
A}, \textbf{70}, 032702, 2004.

\end{thebibliography}
\section{Introduction}

Dynamical systems are one of most extensively studied subjects in physics. Mathematically, a dynamical system can be defined as a quadruplet $(X, \Sigma, µ, \tau)$ where $X$ is a set (typically, the phase space in classical mechanics), $\Sigma$ is a sigma-algebra on $X$, $\mu$ is a finite measure on $\Sigma$ and map $\tau:X\rightarrow X$ is a measure-preserving transformation (see section 2.1). Physical interpretation of this abstract definition is that
a dynamical system gives a fixed rule which describes time dependence of a point (state of system) in a geometrical space (the phase space). This rule is deterministic in the sense that for a given time interval only one future state follows from the current state.

From its origins in Newtonian mechanics to its subsequent measure theoretical definition numerous tools have been developed both theoretical and practical in dynamical systems theory. Some of these are discrete maps, bifurcation theory, topological knots, etc. These give different descriptions\footnote{More precisely, they give local and global descriptions of flux of a dynamical system. Since in this paper we will not use concept of flow, we will not expose any discussion about it. For more details see \cite{lich, guck, gutzwiller}.} of dynamical systems. However, in many cases most important is the asymptotic behavior of  dynamical system \footnote{For example, in synchronized dynamical systems like Kuramoto model (see p. 42 of \cite{kuramoto1}, p. 164 of \cite{kuramoto2}).}(mathematically, $t\rightarrow \infty$). This is the case of approach to equilibrium and we can use tools of chaos theory to study the evolution to equilibrium. Related to this, classical chaos presents several approaches which are related to each other: algorithmic complexity \cite{BELOT}, Lyapunov exponents \cite{lich,guck} and Ergodic Hierarchy \cite{LM,3M}. For instance, Brudno theorem \cite{BELOT} relates complexity with Kolmogorov-Sinai entropy
\cite{BELOT} while Pesin theorem \cite{BELOT} relates Ergodic Hierarchy with Lyapunov exponents. The
relationships between these chaos indicators is illustrated in the ``chaos pyramid" of Fig. 1. According to this structure, Ergodic Hierarchy is
one of features of classical chaos.

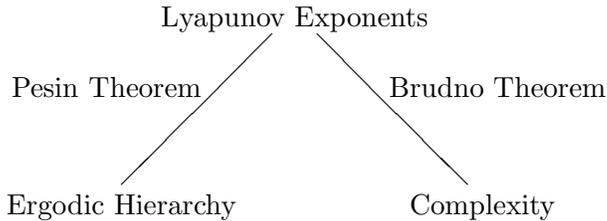
\begin{figure}\label{pyramid}
\begin{center}
\unitlength=1mm
\begin{picture}(5,23)(0,0)
\put(-3,23){\line(-1,-1){20}} \put(3,23){\line(1,-1){20}}
\put(0,25){\makebox(0,0){Lyapunov Exponents}}
\put(-23,0){\makebox(0,0){Ergodic Hierarchy}}
\put(25,0){\makebox(0,0){Complexity}}
\put(27,16){\makebox(0,0){Brudno Theorem}}
\put(-25,16){\makebox(0,0){Pesin Theorem}}
\end{picture}
\caption{``Chaos pyramid" is a diagram for relationships
between Ergodic Hierarchy, Lyapunov exponents and complexity through
Pesin and Brudno theorems.}
\end{center}
\end{figure}

Ergodic Hierarchy ranks the chaotic level of a dynamical system
according to the way in which correlations between two arbitrary
distributions cancel for large times ($t\rightarrow\infty$). For example, in classical mechanics a typical
correlation $C(A,B)$, between two sets A and B of the phase space
$\mathcal{M}$, is $C(A,B)=\mu(A\cap B)-\mu(A)\mu(B)$ where $\mu$ is
a measure defined over subsets of $\mathcal{M}$. If
$\mu(\mathcal{M})=1$ then $\mu(A)$ can be interpreted as
probability of $A$. In such case, $C(A,B)$ can be interpreted as the
difference between probability of $A$ and $B$ simultaneously
and product of probabilities of $A$ and $B$. In
this sense $C(A,B)$ is a correlation that ``measures" how independent are events $A$ and $B$. Moreover, Ergodic Hierarchy ranks
chaotic level of a dynamical system through correlations
$C(A,B)$. For the quantum case, \emph{Quantum Ergodic Hierarchy}
(QEH) \cite{0,NACHOSKY MARIO} also expresses how correlations are
canceled for large times, but in this case quantum correlations
are defined between states and observables as quantum mean values
$(\rho(t)|O)$.

On the other hand, we have an important tool called \emph{Spectral Decomposition Theorem} (SDT) \cite{LM} which describes asymptotic behavior of dynamical systems. Roughly speaking and assuming certain hypotheses about the evolution of system, it says that any density of a dynamical system tends asymptotically ($t\rightarrow \infty$) to a sum of localized densities (more precisely, normalized characteristic functions). As we mentioned above, Ergodic Hierarchy ranks asymptotically chaos between two distributions or densities. So we can expect some kind of connection between Ergodic Hierarchy and Spectral Decomposition Theorem. In fact, a connection is established in Theorem 8 (see section 3) which gives necessary and sufficient conditions for ergodic, mixing and exact levels of Ergodic Hierarchy according to SDT decomposition.

For quantum systems, quantum chaos is a discipline in constant evolution which began in 70's with pioneering experiments of microwave billiards
\cite{stockmann, gutzwiller, haake, casati} and it is well known as the
study of quantum mechanical aspects of quantum systems which have a chaotic classical description. Typical
approaches that emerged from this evolution were scars \cite{stockmann, gutzwiller}, WKB
approximation \cite{stockmann, tabor}, and Random Matrix Theory \cite{stockmann, haake}. A reasonable definition that takes into account the main characteristics of quantum chaos was given by M. Berry \cite{BERRY}: ``\emph{a quantum system is chaotic if its classical limit is chaotic}". This was taken later on to be a possible definition of quantum chaos.

As we mention in introduction of part one \cite{0}, that there are many ways to define quantum chaos by
complexity (see \cite{MANTICA1, MANTICA2}), exponential
divergence trajectories (see \cite{SCHUSTER, BATTERMAN}),
treatment of chaos based on the introduction of non-linear terms in
the Schrodinger equation \cite{WEINBERG} and non-unitary
evolution of a quantum system as an indicator of quantum chaos \cite{GHIRARDI}. However, as mentioned in \cite{0}, the
study of chaos based on QEH takes into
account that classical limit may
not exhibit chaotic behavior which would be a threat to
Correspondence Principle (CP). In paper \cite{FOP} we discuss and give an alternative way to study chaos in quantum systems based on fundamental graininess and the classical statistical limit which is compatible with CP.

As in previous works \cite{0, NACHOSKY MARIO, FOP}, in
this paper we study quantum chaos from
Berry's definition and \emph{Quantum Ergodic Hierarchy} \cite{0}.

The main goal of this paper is a Quantum Version of Spectral Decomposition Theorem (QSDT) which gives a direct connection with
QEH and classical limit. In addition we consider
the degree of generality of QSDT which
make it useful as a framework for quantum chaos. All these aspects are
in accordance with conceptual foundations in context of
Belot-Earman program \cite{BELOT}. Other relevant consequences of QSDT are: A characterization of
first two levels of QEH (see theorem 10 of section 4.2) plus
a simple connection between these levels and
quantum spectrum in the
classical limit (see section 5).
In section 2 we begin introducing a brief review of
minimal notions of density theory for the
development of following sections.

\section{Theory of densities and Markov Operators}

Historically, the concept of density has only recently appeared in
order to unify descriptions of phenomena of statistical
nature. Clear examples of this are Maxwell velocity
distribution and quantum mechanics, considered as attempts to unify
theory of gases and as justification for the derivation of Planck
distribution of black body radiation respectively.

Moreover, development of modern physics demonstrated the usefulness of
densities to give a description of systems with a large number
of freedom degree which have an uncertainty by ignorance. In
this section we introduce a brief review of
theory of densities and Markov operators based on dynamical systems formalism \cite{LM}.

\subsection{Density functions and Dynamical Systems}

We begin recalling mathematical elements of
dynamical systems theory. As we mentioned in introduction, these elements and definitions can be given within the framework of measurement theory\footnote{In this paper this will be our mathematical framework.} \cite{LM, 3M}.
Given a set $X$, $\Sigma$ is a $\sigma$\emph{-algebra} of subsets of $X$ if it
satisfies:

\begin{enumerate}
\item[$(I)$] $X \in \Sigma$

\item[$(II)$] $A,B \in \Sigma \Longrightarrow  A\backslash B \in \Sigma$

\item[$(III)$] $(B_i)\in \Sigma\Longrightarrow \cup_{i}B_i \in \Sigma$ \footnote{Index $i$ must run out a countable set.}
\end{enumerate}

\nd A function $\mu$ on $\Sigma$ is a \emph{probability measure} if
it satisfies:

\begin{enumerate}
\item[$(I)$] $\mu:\Sigma \rightarrow [0,1]$ and $\mu(X)=1$

\item[$(II)$] For all countable family of pairwise disjoint subsets $(B_i)\in \Sigma \Longrightarrow \mu(\cup_{i}B_i )= \sum_{i}\mu(B_i
)$

\end{enumerate}

\nd A \emph{measure space} is any shortlist of the form
$(X,\Sigma,\mu)$. Given a measure space $(X,\Sigma,\mu)$, a
\emph{measure preserving transformation or automorphism} $T$ is a function $T:X\rightarrow X$ which satisfies:

\begin{equation}\label{automorphism}
\forall A\in \Sigma: \mu(T^{-1}A)=\mu(A)
\end{equation}

\nd Then we say that family of transformations
$\tau:=\{T_t:X\rightarrow X\}_{t\in \mathbb{R}}$ satisfying Eq. \eqref{automorphism} and (see definition 7.2.1. of \cite{LM})

\begin{enumerate}
\item[$(a)$] $T_0(x)=x$ for all $x\in X$

\item[$(b)$] $T_{t_1}(T_{t_2}(x))=T_{t_1+t_2}(x)$ for all $x\in X$ and $t_1,t_2\in \mathbb{R}$

\item[$(c)$] The mapping $(t,x)\rightarrow T_t(x)$ from $\mathbb{R}\times X$ into $X$ is continuous.
\end{enumerate}

\noindent is a group\footnote{In the general case $\tau$ is required to be a semigroup, i.e. in semidynamical systems (see definition 7.2.3 of \cite{LM}). In dynamical systems $\tau$ is a group while in semidynamical systems $\tau$ is required to be a semigroup. Therefore, the dynamical systems are invertible while the semidynamical systems may not be invertible.} of measure preserving
automorphisms and we call it a \emph{dynamical law} $\tau$. With
these definitions, we say that quaternary $(X, \Sigma, \mu, \tau)$
is a \emph{dynamical system}.

In dynamical
system theory the central notion is the definition of density: Given a dynamical
system $(X,\Sigma,\mu,\tau)$ and $D(X, \Sigma, \mu )=\{f\in
L^{1}(X, \Sigma, \mu ):f\geq 0\,\ ;\,\ \Vert f\Vert =1\}$ any
function $f \in D(X, \Sigma, \mu, \tau)$ is called a \emph{density}.

In classical mechanics it
is usual to take $X=\mathcal{M}$ as phase space,
$\Sigma=\mathcal{P}(\mathcal{M})$ as power set of phase space,
$\mu$ as Lebesgue measure and $T_t$ as time evolution transformation governed by Hamilton equations. This
context will be clarified in the next sections. In addition to dynamical system definition other fundamental concept is the notion
of Markov operator. This important class of operators is presented
below.

\subsection{Markov Operators}

Given a classical system $S$ with an initial state given by a density
$f_0$ we know that its temporal evolution will be determined by
Liouville equation. Except in simple cases we know that this
equation has no exact solution and thus we are forced to use
another strategy to study the evolution of system. In this
sense Markov operators are very useful because their properties
allow us to know the asymptotic behavior of densities. General behavior of densities can be well developed in both
dynamical and stochastic systems. Markov operators contain global information concerning densities in the asymptotic limit $t\rightarrow \infty$. Under certain hypotheses on Markov operators we have conditions for
existence of an equilibrium density $f_\ast$ which physically corresponds to equilibrium approach. The approach to equilibrium in the limit $t\rightarrow \infty$ by means of Markov operators will be the link between
classical limit and \emph{Quantum Spectral Decomposition
Theorem}. This will be considered in section 4. We present a
brief review of minimal concepts for the development of this paper beginning with the following definition (see
\cite{LM}, pag. 32).

\begin{definition}(\emph{Markov Operator})\label{markov operator}
Given a measure space $(X,\Sigma,\mu)$ a linear operator
$P:L^{1}\rightarrow L^{1} $ is called a \emph{Markov operator} if it
satisfies:

\begin{enumerate}
\item[$(a)$] $Pf \geq 0$

\item[$(b)$] $\|Pf\| = \|f\|$

for all $f\in L^{1} $, $f\geq 0$
\end{enumerate}

\end{definition}

\nd From condition (b) it
follows that $P$ is \emph{monotonic}, that is if $f,g \in L^{1}$
with $f \geq g$ then $Pf \geq Pg$.
Markov operators satisfy important properties that
will be crucial in the derivation of a quantum version of
Spectral Decomposition Theorem (see \cite{LM}, pag. 33).

\begin{theorem}\label{markov properties}
Let $(X,\Sigma,\mu)$ be a $\sigma$-algebra and let $f \in L^{1}$. If
P is a Markov operator then:

\begin{enumerate}

\item[$(I)$] $\|Pf\| \leq \|f\|$ (\emph{contractive} property)

\item[$(II)$] $|Pf(x)| \geq P|f(x)|$
\end{enumerate}
\end{theorem}

\nd The notion of \emph{fixed point} of a Markov operator
is fundamental to establishing the approach to equilibrium of a density (see \cite{LM}, pag. 35).

\begin{definition}(\emph{Fixed Point})\label{fixed point}
Let $P$ be a Markov operator. If $f \in L^{1} $ with $Pf=f$ then $f$
is called a \emph{fixed point} of P. In a more general way, any $f
\in D(X,\Sigma ,\mu)$ that satisfies $Pf = f$ is called a
\emph{stationary density} of P.
\end{definition}
\noindent A family of automorphisms $\{T_t\}_{t\in \mathbb{R}}$ representing the evolution of any dynamical system is a special class of
Markov operators called \emph{Frobenius-Perron operators}. They are
defined as follows (see \cite{LM}, pag. 36).

\begin{definition}(\emph{Frobenius-Perron Operator})\label{perron operator}
Given a measure space $(X,\Sigma,\mu)$ and $T:X\rightarrow X $ a
\emph{non singular automorphism} (i.e. $\mu(T^{-1}(A))=0$ for all $A
\in \Sigma$ such that $\mu(A)=0$) the unique operator
$P:L^{1}\rightarrow L^{1}$ defined for all $A \in \Sigma$ by the
equation

\begin{equation}\label{frobenius}
\int_{A}Pf(x)\mu(dx)=\int_{T^{-1}(A)}f(x)\mu(dx)
\end{equation}
is called the \emph{Frobenius-Perron operator} corresponding to T.
\end{definition}
\nd From Eq. \eqref{frobenius} we have that Frobenius-Perron operator
is linear and has the following properties (see \cite{LM},
pag. 37).

\begin{theorem}\label{perron properties}
Let T be an automorphism. If $P$ and $P_n$ are the Frobenius-Perron
operators corresponding to $T$ and $T^{n}$ respectively. Then we
have that

\begin{enumerate}

\item[$(I)$] $\int_{X}Pf(x)\mu(dx)=\int_{X}f(x)\mu(dx)$

\item[$(II)$] $P_n=P^{n}$

\end{enumerate}
\end{theorem}

\nd The adjoint of a Frobenius-Perron operator is defined
as follows (see \cite{LM}, pag. 42).

\begin{definition}(\emph{Koopman Operator})\label{perron operator}
Given a measure space $(X,\Sigma,\mu)$ and $T:X\rightarrow X $ a
\emph{non singular automorphism}, the unique operator
$U:L^{\infty}\rightarrow L^{\infty}$ defined for all $f \in
L^{\infty}$ by

\begin{equation}\label{koopman}
Uf(x)=f(T(x))
\end{equation}
is called the \emph{Koopman operator} corresponding to T.
\end{definition}

\noindent From Eq. \eqref{koopman} we have that Koopman operator is
linear and has the following properties (see \cite{LM}, pag.
43):

\begin{theorem}\label{Koopman properties}
Let T be an automorphism. If $U$ and $U_n$ are the Koopman operators
corresponding to $T$ and $T^{n}$ respectively. Then

\begin{enumerate}

\item[$(I)$] $\|Uf\|_{L^{\infty}}\leq \|f\|_{L^{\infty}}$

\item[$(II)$] $U_n=U^{n}$

\item[$(III)$] $\langle P_n f, g\rangle = \langle f, U_n g\rangle$ \,\,\,\,\,\ for every $f \in L^{1}, g \in L^{\infty}$, $n \in \mathbb{N}_0$

\end{enumerate}

where $\langle f,g\rangle=\int_{X}f(x)g(x)\mu(dx)$ for all $f
\in L^{1}, g \in L^{\infty}$.
\end{theorem}

\nd Now we recall the \emph{Ergodic Hierarchy} of dynamical systems
(see \cite{LM}, pag. 66).

\begin{theorem}(\emph{Ergodic, Mixing and Exact})\label{EH}
Let $(X,\Sigma,\mu)$ be a normalized measure space, $T:X \rightarrow
X$ an automorphism and $P$, $U$ the Frobenius-Perron and Koopman
operators corresponding to $T$. Then:

\begin{enumerate}

\item[$(a)$] $T$ is \emph{ergodic} $\Leftrightarrow$ $lim_{n\rightarrow \infty} \frac{1}{n}\sum_{k=0}^{n}\langle P^{k} f, g\rangle
= lim_{n\rightarrow \infty} \frac{1}{n}\sum_{k=0}^{n} \langle f, U_k
g\rangle = \langle f,1\rangle \langle 1,g \rangle$
\,\,\,\,\,\,\,\,\,\,\,\,\,\,\,\ for all $f \in L^{1}, g \in
L^{\infty}$

\item[$(b)$] $T$ is \emph{mixing} $\Leftrightarrow$ $lim_{n\rightarrow \infty} \langle P^{n} f, g\rangle
= lim_{n\rightarrow \infty} \langle f, U_n g\rangle = \langle
f,1\rangle \langle 1,g \rangle$
\,\,\,\,\,\,\,\,\,\,\,\,\,\,\,\,\,\,\,\,\,\,\,\,\,\,\,\,\,\,\,\,\,\,\,\,\,\,\,\,\,\,\,\,\,\,\,\,\,\,\,\,\,\,\,\,\,\,\,\,\
for all $f \in L^{1}, g \in L^{\infty}$

\item[$(c)$] $T$ is \emph{exact} $\Leftrightarrow lim_{n\rightarrow \infty} \| P^{n} f - \langle f,1\rangle \|= 0$ \,\,\,\,\,\ for all $f \in
L^{1}$

\end{enumerate}

\end{theorem}

\nd From these definitions (see \cite{LM} pag. 73) it follows that \footnote{Ergodic, mixing and exact correspond to C\`{e}saro limit, weak limit and strong limit respectively. The terms ``C\`{e}saro", ``weak" and ``strong" indicate the type of convergence of sequence $\{P^{n} f\}$ to $1$ (see theorem 4.4.1 of \cite{LM}).}

\begin{equation}
EXACT \,\ \Longrightarrow \,\ MIXING \,\ \Longrightarrow \,\ ERGODIC
\end{equation}

\nd It should be noted that levels are substantially different from each other \footnote{A full distinction between ergodic, mixing and exact is illustrated by first six successive iterates of a random distribution of $1000$ points in $X=[0,1]\times[0,1]$ with the corresponding transformations $S_{ergodic}(x,y)=(\sqrt{2}+x,\sqrt{3}+y)$ (mod 1), $S_{mixing}(x,y)=(x+y,x+2y)$ (mod 1) and $S_{exact}(x,y)=(3x+y,x+3y)$ (mod 1).
In each case the effect of transformation is to move around the space, spread throughout space and quickly spread throughout the space corresponding to the ergodic, mixing and exact levels respectively (see fig. 4.3.3, fig. 4.3.4 and fig. 4.3.5 of \cite{LM}).}.
To end this section we introduce the notion of constrictive
operator which allow us to ensure the existence of an equilibrium
density (see \cite{LM}, pag. 87).

\begin{definition}(\emph{constrictive operator})\label{constrictive}
A Markov operator P will be called \emph{constrictive} if there
exist a precompact set $\mathcal{F}\subseteq L^{1}$ such that for
all $f \in D(X,\Sigma ,\mu)$:

\begin{equation}
lim_{n\rightarrow \infty} d(P^{n}f,\mathcal{F}) = lim_{n\rightarrow
\infty} inf _{g \in \mathcal{F}}\|P^{n}f - g \| = 0
\end{equation}
\end{definition}

\nd A relevant result is that every Markov constrictive operator
has an equilibrium density (see \cite{LM}, pag. 87).

\begin{theorem}\label{constrictive MO}
Let $(X,\Sigma,\mu)$ be a normalized measure space, and
$P:L^{1}\rightarrow L^{1}$ a constrictive Markov operator. Then P
has a stationary density, i.e there is a $f_{\ast}\in L^{1}$ such
that $Pf_{\ast} = f_{\ast}$.
\end{theorem}
\nd The existence of equilibrium densities can be considered relevant to obtain a theoretical framework for quantum chaos. For instance, in quantum billiards the statement of quantum ergodicity is equivalent to the statement of uniform equidistribution of probability density for eigenstates. Moreover, this fact can be
translated to quantum language through Wigner transform such
that equilibrium state of any closed quantum system with continuous spectrum is represented by the weak limit $\hat{\rho}_{\ast}$
(see \cite{MARIO OLIMPIA}, pag. 889 eq. (3.28)).

\section{The Spectral Decomposition Theorem of Dynamical Systems}

With all the mathematical background of previous section we are
able to present one of the main results of dynamical
systems theory called the \emph{Spectral Decomposition Theorem} (SDT)
(see \cite{LM}, pag. 88).

\begin{theorem}(\emph{The Spectral Decomposition Theorem}\emph{(version I)})\label{SDT2}
Let P be a Markov constrictive operator. Then there is an integer r,
two sequences of nonnegative functions $g_i \in D(X,\Sigma ,\mu)$,
$k_i \in L^{\infty}$, $i=1,...,r$, and an operator
$Q:L^{1}\rightarrow L^{1}$ such that for all $f \in L^{1}$, $Pf$ may
be written as
\begin{equation}\label{SDT1 decomposition}
Pf(x)=\sum_{i=1}^{r}\lambda_i(f)g_i(x) + Qf(x)
\end{equation}

\nd where

\begin{equation}\label{lambda}
\lambda_i(f)= \int_{X} f(x)k_i(x)\mu(dx)=\langle f(x),k_i(x)\rangle
\end{equation}

\nd The functions $g_i$ and the operator $Q$ have the following
properties:

\begin{enumerate}

\item[$(I)$] $g_i(x)g_j(x)=0$ for all $i\neq j$, so that functions $g_i$ have disjoint supports.

\item[$(II)$] For each integer $i$ there exists a unique integer $\alpha (i)$ such that $Pg_i=g_{\alpha (i)}$. Further $\alpha (i) \neq \alpha (j)$ for $i \neq j$ thus the operator $P$ just permutes the functions $g_i$.

\item[$(III)$] $\|P^{n}Qf\| \rightarrow 0$ \,\,\, as \,\,\, $n \rightarrow \infty$ for every $f \in L^{1}$.

\end{enumerate}
\end{theorem}

\nd Assuming that time
evolution has a constrictive Perron-Frobenius operator then Eq. \eqref{SDT2} describes the evolution of any
density having an initial term which oscillates
between densities $g_i$. We note that $Qf$ is associated with a relaxation process which we will analyze in
section 5. By property (II) of Theorem \eqref{SDT2} we have

\begin{equation}
P^{n}f(x)=\sum_{i=1}^{r}\lambda_i(f)g_{\alpha^{n}(i)}(x) +
P^{n-1}Qf(x)=\sum_{i=1}^{r}\lambda_{\alpha^{-n}(i)}(f)g_i(x) +
Q_{n}f(x)
\end{equation}

\nd where $Q_{n}f(x)=P^{n-1}Qf(x)$ and $\{\alpha^{-n}(i)\}$ is the
inverse permutation of $\{\alpha^{n}(i)\}$.
When the measure space is normalized and Markov
operator $P$ has a constant stationary density $f_{\ast}$ (e.g., if P
is a Frobenius-Perron operator this is equivalent to $\mu_f$
invariant, see \cite{LM}, pag. 46) Spectral Decomposition
Theorem takes the following compact form (see \cite{LM}, pag. 90).

\begin{theorem}(\emph{Version (II) of The Spectral Decomposition Theorem})\label{SDT}
Let $(X,\Sigma ,\mu)$ be a normalized measure space and
$P:L^{1}\rightarrow L^{1}$ a constrictive Markov operator. If P has
a stationary density then the representation of $P^{n}f$ takes the
simple form for all $f \in L^{1}$

\begin{equation}\label{SDT2 decomposition}
P^{n}f(x)=\sum_{i=1}^{r}\lambda_{\alpha^{-n}(i)}(f)\overline{1}_{A_i}(x)
+ Q_{n}f(x)
\end{equation}

\nd where

\begin{equation}\label{characteristic1}
\overline{1}_{A_i}(x)=[1/\mu(A_i)]1_{A_i}
\end{equation}

\begin{equation}\label{characteristic1bis}
\bigcup_{i}A_i=X \,\,\,\ with \,\,\,\ A_i \cap A_j = \emptyset
\,\,\,\ for \,\,\,\ i\neq j
\end{equation}

\nd and \footnote{In equations (9)-(12) the index $i$ runs over a countable set.}

\begin{equation}
\mu(A_i) = \mu(A_j) \,\,\,\ if \,\,\,\ j=\alpha^{n}(i) \,\,\,\ for
\,\,\,\ some\,\,\,\ n.
\end{equation}

\end{theorem}

\nd Theorem \eqref{SDT}
characterizes the Ergodic Hierarchy levels by means of permutation
$\{\alpha^{n}(i)\}$ (see \cite{LM}, pag. 92 - 94).

\begin{theorem}(\emph{The Spectral Decomposition Theorem and Ergodic Hierarchy})\label{SDT and EH}
Let $(X,\Sigma ,\mu)$ be a normalized measure space and
$P:L^{1}\rightarrow L^{1}$ a constrictive Markov operator. Then

\begin{enumerate}

\item[$(I)$] $P$ is \emph{ergodic} $\Longleftrightarrow$ permutation $\{\alpha(1),...,\alpha(r)\}$ of sequence $\{1,...,r\}$ is \emph{cyclical} (that is, for which there no is invariant subset).

\item[$(II)$] If $r=1$ in representation of Eq. $(6)$ $\Longrightarrow$ $P$ is \emph{exact}.

\item[$(III)$] If $P$ is \emph{mixing} $\Longrightarrow$ $r=1$ in representation of Eq. $(6)$.

\end{enumerate}

\nd It should be noted that, under hypothesis of Theorem 8, $P$ mixing implies $P$ exact.
Then since $P$ exact implies $P$ mixing (see eq. (4)) it follows that $P$ exact $\Longleftrightarrow$ $P$ mixing $\Longleftrightarrow$ $r=1$ (see remark 5.5.1 of \cite{LM}).

\end{theorem}
\nd Theorem \eqref{SDT} says that
if classical Hamiltonian is such that Frobenius-Perron operator
$P$ associated to time evolution $T$ admits an equilibrium
density $f_{\ast}$ then in the asymptotic limit ($n \rightarrow \infty$)
the state of the system $U_nf(x)=f(T_n(x))$ will oscillate between
characteristic functions $\overline{1}_{A_i}(x)$ with a remainder
term $Q_{n}f(x)$ going to zero. In next section we will see that a quantum version of SDT (called QSDT)
contains relevant information about quantum spectrum. We have seen that
constrictiveness of a Markov operator $P$ and
normalization of measure space are sufficient to ensure the
existence of stationary densities and to obtain a representation of
their time evolution by means of SDT.

\section{The Quantum Version of Spectral Decomposition Theorem (QSDT)}

The aim of this paper is to obtain a quantum version of SDT, called the \emph{Quantum Version of Spectral Decomposition Theorem} (QSDT), that can be useful to study quantum
systems in the classical limit and to give an alternative framework for quantum chaos according to Berry's
definition. In this section we begin by defining the mathematical
concepts of this approach considering the observables as central
objects and the states as functionals of those. In order to apply QSDT to open quantum systems (see section 6) we consider the quantum characteristic algebra $\mathcal{A}$ whose elements
$\hat{O}\in\mathcal{A}$ are observables not necessarily self-adjoint. In other words, Hermiticity condition $\hat{O}=\hat{O}^{\dag}$ can be relaxed and this is motivated for several reasons some of them we can list below:
\begin{enumerate}
\item[(a)] The interest in the study of non-Hermitian Hamiltonians related with the interpretation of some properties such as transfer phenomena, nuclear resonances, typical of open systems.
\item[(b)] In scattering systems one can consider quantum resonances, i.e. ``quasi-stationary states" instead of scattering solutions \cite{sieber}. These resonances can play a similar role in open systems as eigenstates in closed systems and their eigenvalues are complex numbers with non zero imaginary part, see sections 5.1.1. and 6.
\item[(c)] Any measurement on a wave system drastically changes its properties by converting discrete energy levels into decaying resonances called ``quasi-stationary states" which can be described by a non-Hermitian Hamiltonian \cite{kuhl, rotten, moiseyev}. This situation will be illustrated in section 6.
\item[(d)] With the purpose to take into account non-unitary time evolutions that appear in descriptions of open quantum systems the introduction of non-Hermitian observables becomes ``natural" and fundamental from a theoretical viewpoint.
\item[(e)] Non-Hermitian operators are frequently used to mathematically represent potentials of open quantum systems describing ionization or dissociation where the system breaks up into freely moving non-interacting subsystems, see pag. 4 of \cite{moiseyev}.
\item[(f)] The complex expectation value $A_{\hat{O}}=\langle \hat{O}\rangle=|A_{\hat{O}}|e^{i\alpha}$ of a non-Hermitian observable $\hat{O}$ can be physically interpreted postulating that the absolute value $|A_{\hat{O}}|$ and the phase $\alpha$ are measurable quantities, see pag. 11 of \cite{moiseyev}.
\end{enumerate}
We point out that conclusions obtained from QSDT will be valid for both closed and open quantum systems. However, in section 6 we will only illustrate QSDT with examples of open quantum systems in the research line\footnote{We mean that examples presented in this paper belong theoretically to the category of quantum systems within the approach of \cite{zycz}, i.e. open quantum systems which are disturbed by the measuring process like microwave billiards, etc.} of \cite{zycz} which states: \emph{``The approach of linking chaos with the unpredictability of the measurement outcomes is the right one in the quantum case"}. In this paper we will not discuss this entropic approach of quantum chaos but it should be noted that\footnote{The term ``intrinsic" refers to quantum dynamical of closed quantum systems described by wave functions of integrable square and thus quantum evolution becomes almost periodic at least for the finite dimensional case. In contrast, the term ``genuine" is associated with the approach based on that successive outcomes of a measurement can form chaotic and unpredictable sequences \cite{zycz}.}

\begin{equation}\label{closed open}
\begin{split}
&closed \,\ quantum \,\ systems \Longrightarrow hermitian \,\ observables \Longrightarrow ``intrinsic \,\ quantum \,\ chaos"\\
&\\
&open\,\ quantum \,\ systems \Longrightarrow observables \,\ (herm. \,\ and \,\ non-herm.) \Longrightarrow ``genuine\,\ quantum \,\ chaos"
\end{split}
\end{equation}
\nd This distinction can be seem trivial and meaningless but it allows to place algebra of hermitian operators and unitary evolutions in the context of closed quantum systems. On the other hand, second line of Eq. \eqref{closed open} places algebra of operators (hermitian and non-hermitian) and evolutions (unitary and non-unitary) in the context of open quantum systems.
Returning to the mathematical background, the space of states is the positive cone

\begin{equation}
\mathcal{N}=\{\hat{\rho} \in
\mathcal{A}^{\prime}:\hat{\rho}(\mathbb{I})=1, \,\,\
\hat{\rho}^{\dag}=\hat{\rho}, \,\,\
\hat{\rho}(\hat{a}.\hat{a}^{\dag})\geq 0 \,\,\ for \,\,\ all \,\,\
\hat{a} \in \mathcal{A}\}
\end{equation}

\nd where the action $\hat{\rho}(\hat{O})$ of functional
$\hat{\rho} \in \mathcal{A}^{\prime}$ on the observable $\hat{O} \in
\mathcal{A}$ is denoted by $(\hat{\rho}|\hat{O})$. In other words, $(\hat{\rho}|\hat{O})$ is the mean value of $\hat{O}$ in $\hat{\rho}$, i.e.
$(\hat{\rho}|\hat{O})=\langle\hat{O}\rangle_{\hat{\rho}}=tr(\hat{\rho}\hat{O})$. When $\hat{O}=\mathbb{I}$ this action is the trace of $\hat{\rho}$ which is equal to
$tr(\hat{\rho})=\hat{\rho}(\mathbb{I})=(\hat{\rho}|\mathbb{I})=1$. We note the dual of $\mathcal{A}$ as $\mathcal{A}^{\prime}$.

In this approach the state is unknown and we only focus in the study of
expectation values $(\hat{\rho}(t)|\hat{O})$ for large times $t \rightarrow
\infty$. If for all $\hat{\rho} \in \mathcal{A}^{\prime}$ there exists a unique $\hat{\rho}_{\ast} \in
\mathcal{A}^{\prime}$ such that

\begin{equation}
lim_{t \rightarrow \infty} (\hat{\rho}(t)|\hat{O}) = lim_{t
\rightarrow \infty} (\hat{U}_t \hat{\rho}\hat{U}_t^{\dag}|\hat{O}) =
(\hat{\rho}_{\ast}|\hat{O})
\end{equation}

\nd we say that $\hat{\rho}$ has \emph{weak-limit}
$\hat{\rho}_{\ast}$ (see \cite{0} pag. 248). Functional
$\hat{\rho}_{\ast}$ is interpreted as the average value that would
result if state $\hat{\rho}(t)$ had a limit $\hat{\rho}_{\ast}$
for $t \rightarrow
\infty$. That is, $\hat{\rho}_{\ast}$ is a weak limit and not a limit in the
(strong) usual sense. In other words $\hat{\rho}_{*}$ is an equilibrium state, in the weak sense, that system reaches in its process of relaxation.

Now it is suitable to make the following remark. In general, if evolution operator $\hat{U_t}$ is not unitary (for example when Hamiltonian is non-Hermitian, see section 6) then we will have that $0<tr(\hat{\rho}_{\ast})<1$, i.e. weak limit $\hat{\rho}_{\ast}\in\mathcal{A}^{\prime}\backslash\mathcal{N}$ is considered an state not normalized. This means that trace of any initial state $\hat{\rho}$ is not preserved for non-unitary evolutions in the limit $t\rightarrow\infty$. This is a typical feature of decoherence in open quantum systems where interaction between a quantum system and its environment produces a non-unitary evolution of  reduced state described by a master equation \cite{zurek-paz}.

A fundamental mathematical element to study the
classical limit of a quantum system is the
\emph{Wigner transform}. From a quantum state $\hat{\rho}$ Wigner transform allows us to obtain a
function $f(\phi)$ defined over phase space $\Gamma$ that can be
interpreted, in the classical limit $\hbar\rightarrow0$, as a distribution probability governed by Liouville
equation of classical statistical mechanics. We present a brief review of some relevant properties of Wigner
transform.

Let $\Gamma=\mathcal{M}_{2(N+1)}\equiv \mathbb{R}^{2(N+1)}$ be the
phase space. Wigner transformation
$symb:\mathcal{A}\rightarrow \mathcal{A}_q$ sends
quantum algebra $\mathcal{A}$ to a ``classical-like"
$\mathcal{A}_q$ algebra by (see \cite{HILLERY, DITO, GADELLA}).

\begin{equation}
symb(\hat{f})=f(\phi)=\int_{\Gamma} \langle q+\Delta|\hat{f}|q-\Delta\rangle
e^{i\frac{p\Delta}{\hbar}}d^{N+1}\Delta
\end{equation}

\nd where $f(\phi)\in \mathcal{A}_q$ are functions defined over
space phase $\Gamma$ where
$\phi=(q^{1},...,q^{N+1},p_{q}^{1},...,p_{q}^{N+1})$.
The star product between two operators
$\hat{f},\widehat{g}\in\mathcal{A}$ is given by (see \cite{HILLERY})

\begin{equation}\label{star product}
symb(\hat{f}.\hat{g})=symb(\hat{f})\ast symb(\hat{g})=(f\ast
g)(\phi)
\end{equation}

\nd and the \emph{Moyal bracket} (see \cite{GADELLA}) is

\begin{equation}\label{moyal}
\{f,g\}_{mb}=\frac{1}{i\hbar}(symb(\hat{f})\ast
symb(\hat{g})-symb(\hat{g})\ast
symb(\hat{f}))=symb(\frac{1}{i\hbar}[\hat{f},\hat{g}])
\end{equation}

\nd Two important properties are, see \cite{HILLERY}
\begin{equation}\label{moyal2}
(f\ast g)(\phi)=f(\phi)g(\phi)+0(\hbar) \,\,\,\,\,\,\, ,
\,\,\,\,\,\,\, \{f,g\}_{mb}=\{f,g\}_{pb}+0(\hbar^{2})
\end{equation}

\nd The \emph{symmetrical} or Weyl ordering prescription is used to
define inverse map \emph{symb}$^{-1}$, that is

\begin{equation}
symb^{-1}[q^{i}(\phi),p^{j}(\phi)]=\frac{1}{2}(\hat{q}^{i}\hat{p}^{j}+\hat{p}^{j}\hat{q}^{i})
\end{equation}

\nd Therefore, \emph{symb} and \emph{symb}$^{-1}$ define
isomorphisms between the algebras $\mathcal{A}$ and
$\mathcal{A}_q$,

\begin{equation}
symb:\mathcal{A}\rightarrow \mathcal{A}_q \,\,\,\,\,\,\, ,
\,\,\,\,\,\,\, symb^{-1}:\mathcal{A}_q\rightarrow
\mathcal{A}
\end{equation}

\nd On the other hand, Wigner transformation for states is

\begin{equation}
\rho(\phi)=(2\pi\hbar)^{-(N+1)}symb(\hat{\rho})
\end{equation}

\nd The fundamental property of Wigner transform used throughout
this paper is the preservation of inner product between states
$\hat{\rho}\in \mathcal{N}$ and observables
$\hat{O}\in\mathcal{A}$. From a physical viewpoint this property correspond to the invariance of mean values calculated
in $\mathcal{A}$ and $\mathcal{A}_q$ respectively. More precisely,

\begin{equation}\label{wigner}
\hat{\rho}(\hat{O})=\langle
\hat{O}\rangle_{\hat{\rho}}=(\hat{\rho}|\hat{O})=(symb(\hat{\rho})|symb(\hat{O}))=\langle\rho(\phi),O(\phi)\rangle=\int_{\Gamma}
d\phi^{2(N+1)}\rho(\phi)O(\phi)
\end{equation}
\nd More generally, if $\hat{A},\hat{B}\in\mathcal{A}$ then we have (see \cite{HILLERY, case})

\begin{equation}\label{wigner2}
tr(\hat{A}\hat{B})=(\hat{A}|\hat{B})=(symb(\hat{A})|symb(\hat{B}))=\langle A(\phi),B(\phi)\rangle=\int_{\Gamma}
d\phi^{2(N+1)}A(\phi)B(\phi)
\end{equation}
\nd where $A(\phi)=(2\pi\hbar)^{-(N+1)}symb(\hat{A})$ and $B(\phi)=symb(\hat{B})$.

\nd At this point it is suitable to make the following remarks and clarifications.
\begin{enumerate}

\item[$(I)$] In this paper we use $\hat{\rho}(n)$ as a short notation of $\hat{\rho}$ after $n$ successive applications of $\hat{U}$, i.e. $\hat{\rho}(n)=\hat{U}(n) \hat{\rho} \hat{U}(n)^{\dag}$ represents $\hat{\rho}$ after $n$ arbitrary time steps at constant intervals in a discretized time evolution. More precisely, we consider that $\hat{U}(n)$ is the one given by Hamiltonian or Floquet operator \cite{stockmann}. Both cases are of great interest theoretically and experimentally\footnote{For instance, when we have Hamiltonians with periodic time dependences the theoreticians prefer periodically \emph{kicked} systems while the experimentalists prefer \emph{driven} systems (see pag. 138 and 139 of \cite{stockmann}).}. That is, we can choose $\hat{U}$ as evolution operator and make the time step to be equal to $\alpha\in\mathbb{R}$ then we have $\hat{\rho}(n)=(e^{-i\frac{\hat{H}}{\hbar}\alpha n}) \hat{\rho} (e^{-i\frac{\hat{H}^{\dag}}{\hbar}\alpha n})$ where only in the unitary case we will have $\hat{H}^{\dag}=\hat{H}$. Moreover, in applications of oscillating electric (or magnetic) fields we have Hamiltonians of the type $\hat{H}=\hat{H}_0+\hat{V}(t)$ where $\hat{H}_0$ is typically the Hamiltonian of an atom or nucleus with a time-periodic potential $\hat{V}(t)$. In this case the ``natural" choice for $\hat{U}$ is the Floquet operator so that $\hat{\rho}(n)$ represents an stroboscopic observation of the system in  $\hat{\rho}$ at time $t=n\tau$ where $\tau$ is the periodicity of $\hat{V}(t)$ (see section 4.1 of \cite{stockmann}). In section 5 we will take into account that the choice of time steps is arbitrary.

\item[$(II)$] By ``a time evolution operator $\hat{U}$ having a classical evolution automorphism $T$" we will mean that if $\hat{\rho}$ is a quantum state and $\rho=symb(\hat{\rho})$ is its corresponding density then Wigner transform of $\hat{\rho}(1)$ (i.e. $\hat{\rho}(1)$ is $\hat{\rho}$ after an application of $\hat{U}$) is $symb(\hat{\rho}(1))=symb(\hat{U}\hat{\rho}\hat{U}^{\dag})=\rho \circ T$. In other words, Wigner transform connects evolution operator $\hat{U}$ with automorphism $T$. Moreover, it follows that $symb(\hat{\rho}(n))=symb(\hat{U}(n)\hat{\rho}\hat{U}(n)^{\dag})=\rho \circ T^{n}$ for all $n\in \mathbb{N}\cup \{0\}$.
\end{enumerate}

\noindent These remarks will be fundamental for the development of next sections.
\subsection{The Quantum Spectral Decomposition Theorem (QSDT)}

In previous section we presented a framework based on the
classical limit by means of Wigner transform. Now we can write
Spectral Decomposition Theorem (Theorem 7 of section 3) in quantum language. First, we assume
that

\begin{itemize}
\item If $\hat{U}$ represents the evolution\footnote{By remark $(I)$ of previous section we know that $\hat{U}$
 is not necessarily $\hat{U}=e^{-i\frac{\hat{H}}{\hbar}t}$ given by Hamiltonian, i.e. $\hat{U}$ may also be Floquet operator.} of a quantum system then $\hat{U}$
has a corresponding classical evolution automorphism \footnote{Of course, we can make the
natural choice of $T$ as $T:\Gamma\rightarrow \Gamma$ with
$T(\phi=(q,p))=\phi(1)=(q(1),p(1))$, i.e. time step is equal to one and $T$ is the classical
evolution operator determined by Hamilton equations. However, the choice of time steps is arbitrary.} $T$ defined
over $\Gamma$ and $T$ has an associated constrictive Frobenius-Perron
operator $P$.

\item There exists a stationary density $f_{\ast}$, that is,
$Pf_{\ast}=f_{\ast}$.
\end{itemize}

\noindent Under these hypothesis we have the following quantum version of
Spectral Decomposition Theorem.

\begin{theorem}(\emph{The Quantum Spectral Decomposition Theorem}\emph{(QSDT)})\label{QSDT}
Let $\hat{\rho} \in \mathcal{N}$ and let $\hat{O}$ be an observable.
Then there exists pure states
$\hat{\rho}_1,\hat{\rho}_2,...,\hat{\rho}_r$; observables
$\hat{O}_1,\hat{O}_2,...,\hat{O}_r$ ; a permutation
$\alpha:\{1,...,r\}\longrightarrow \{1,...,r\}$ and
$\widetilde{\rho}_0 \in \mathcal{A}^{\prime}$ such that

\begin{equation}\label{QSDT1}
(\hat{\rho}(n)|\hat{O})=\sum_{i=1}^{r}\lambda_{\alpha^{-n}(i)}(\hat{\rho}_i|\hat{O})
+ (\widetilde{\rho}_0(n-1)|\hat{O})
\end{equation}

\nd where\footnote{$\hat{\rho}(n)$ in the sense of the remark $(I)$ of pag. 11.}

\begin{equation}\label{QSDT2}
\begin{split}
&\hat{\rho}(n)=U(n) \hat{\rho} U(n)^{\dag}   \,\,\,\,\,\, and \\
&\\
&\lambda_i(\hat{\rho})= (\hat{\rho}|\hat{O}_i)
\end{split}
\end{equation}

\nd The states $\hat{\rho}_i$ and $\widetilde{\rho}_0$ have the
following properties:

\begin{enumerate}

\item[$(I)$] $\hat{\rho}_i\hat{\rho}_j=0(\hbar)$ for all $i\neq j$ and $\hat{\rho}_i^{2}=\hat{\rho}_i+0(\hbar)$. So that states $\hat{\rho}_i$ are projectors in the classical limit $(\hbar\rightarrow
0)$. Moreover, we have a decomposition of the identity:

\begin{equation}\label{QSDT3}
\hat{1} = \sum_{i}\alpha_i\hat{\rho}_i \,\,\,\,\,\,\,\ with
\,\,\,\,\,\,\,\ \alpha_i\geq0 \,\,\,,\,\,\sum_{i}\alpha_i=1
\end{equation}

\item[$(II)$] For each integer $i$ there exists a unique integer $\alpha (i)$ such that $(\hat{U}\hat{\rho}_i\hat{U}^{\dag}|\hat{O})=(\hat{\rho}_{\alpha (i)}|\hat{O})$. Further $\alpha (i) \neq \alpha (j)$ for $i \neq j$ so operator $\hat{U}$ permutes the states $\hat{\rho}_i$.

\item[$(III)$] $(\widetilde{\rho}_0(n-1)|\hat{O}) \longrightarrow 0$ \,\,\, as \,\,\, $n \longrightarrow \infty$.

\end{enumerate}
\end{theorem}

\begin{proof}
We consider that triplet $(X,\Sigma,\mu)$ is the one given by classical mechanics, i.e $X=\Gamma$ is the $2(N+1)$-dimensional phase space, $\Sigma=\mathcal{P}(\Gamma)$ is the power set of $\Gamma$ and $\mu$ is the Lebesgue measure. As is usual, we denote $x\in X$ by $\phi=(q^{1},...,q^{N+1},p_{q}^{1},...,p_{q}^{N+1})\in\Gamma$ and $\mu(dx)$ by $d^{2(N+1)}\phi$. Let $\hat{\rho} \in \mathcal{N}$ and let $\hat{O}$ be an observable.
If we define $f=symb(\hat{\rho})$ and $g=symb(\hat{O})$ then,
multiplying Eq. (9) by $g$ and integrating over all phase space $\Gamma$ we
have

\begin{equation}
\int_{\Gamma}d^{2(N+1)}\phi \ P^{n}f(\phi)g(\phi)=\sum_{i=1}^{r}\lambda_{\alpha^{-n}(i)}(f)\int_{\Gamma}d^{2(N+1)}\phi \ \overline{1}_{A_i}(\phi)g(\phi)
+ \int_{\Gamma}d^{2(N+1)}\phi \ Q_{n}f(\phi)g(\phi)
\end{equation}

\nd Equivalently,

\begin{equation}
\langle P^{n}f,g\rangle =\sum_{i=1}^{r}\lambda_{\alpha^{-n}(i)}(f)
\langle\overline{1}_{A_i},g\rangle + \langle P^{n-1}Qf,g\rangle
\end{equation}

\nd Since $\langle P^{n}f,g\rangle = \langle f,U^{n}g\rangle$ and
$\langle P^{n-1}Qf,g\rangle = \langle Qf,U^{n-1}g\rangle$ (i.e.
Koopman operator $U$ is the dual of Frobenius-Perron
operator $P$ both corresponding to automorphism $T$) we obtain\footnote{Throughout the demonstration we consider that $T$ is an arbitrary automorphism and not necessarily the classical evolution $T_{t=1}(\phi=(q,p))=\phi(1)=(q(1),p(1))$ given by the Hamiltonian equations.}
\begin{equation}
\langle f,U^{n}g\rangle =\sum_{i=1}^{r}\lambda_{\alpha^{-n}(i)}(f)
\langle\overline{1}_{A_i},g\rangle + \langle Qf,U^{n-1}g\rangle
\end{equation}

\nd Now if we call $\hat{\rho}_i=symb^{-1}(\overline{1}_{A_i})$,
$\widetilde{\rho}_0=symb^{-1}(Qf)$ and use that
$U^{n}g=g\circ T^n=g(n)=symb(\hat{O}(n))$, $U^{n-1}g=g\circ T^{n-1}=g(n-1)=symb(\hat{O}(n-1))$
(see Eq. \eqref{koopman}) then

\begin{equation}\label{symb1}
\langle symb(\hat{\rho}),symb(\hat{O}(n))\rangle
=\sum_{i=1}^{r}\lambda_{\alpha^{-n}(i)}(f) \langle
symb(\hat{\rho}_i),symb(\hat{O})\rangle + \langle
symb(\widetilde{\rho}_0),symb(\hat{O}(n-1))\rangle
\end{equation}

\nd If we call $k_i=symb(\hat{O}_i)$ and use Eq.
\eqref{lambda} then coefficient $\lambda_{\alpha^{-n}(i)}(f)$
can be written as

\begin{equation}
\lambda_{\alpha^{-n}(i)}(f)=\int_{\Gamma}
d^{2(N+1)}\phi \ f(\phi)k_{\alpha^{-n}(i)}(\phi)=\langle
f,k_{\alpha^{-n}(i)}\rangle=\langle
symb(\hat{\rho}),symb(\hat{O}_{\alpha^{-n}(i)})\rangle=\lambda_{\alpha^{-n}(i)}(\hat{\rho})
\end{equation}

\nd Therefore, Eq. \eqref{symb1} reads

\begin{equation}\label{symb2}
\langle symb(\hat{\rho}),
symb(\hat{O}(n))\rangle=\sum_{i=1}^{r}\lambda_{\alpha^{-n}(i)}(\hat{\rho})\langle
symb(\hat{\rho}_i), symb(\hat{O})\rangle + \langle
symb(\widetilde{\rho}_0), symb(\hat{O}(n-1))\rangle
\end{equation}

\nd Finally, we can use the property of Wigner transform given by Eq. \eqref{wigner}
\begin{equation}\label{WEYL}
\forall \hat{O}\in \mathcal{A},\forall\hat{\rho}\in
\mathcal{A}^{\prime}:(\hat{\rho}|\hat{O})=\langle
symb(\hat{\rho}), symb(\hat{O})\rangle=\int_{\Gamma}d^{2(N+1)}\phi\ \rho(\phi)O(\phi)
\end{equation}

\nd Using this property Eq. \eqref{symb2} can be expressed
in quantum language as

\begin{equation}\label{symb3}
(\hat{\rho}|\hat{O}(n))=\sum_{i=1}^{r}\lambda_{\alpha^{-n}(i)}(\hat{\rho})(\hat{\rho}_i|\hat{O})
+ (\widetilde{\rho}_0|\hat{O}(n-1))
\end{equation}

\nd We know that $(\hat{\rho}|\hat{O}(n))$ and
$(\widetilde{\rho}_0|\hat{O}(n-1))$ are equal to
$(\hat{\rho}(n)|\hat{O})$ and $(\widetilde{\rho}_0(n-1)|\hat{O})$
respectively, and then Eq. \eqref{symb3} reads as

\begin{equation}
(\hat{\rho}(n)|\hat{O})=\sum_{i=1}^{r}\lambda_{\alpha^{-n}(i)}(\hat{\rho})(\hat{\rho}_i|\hat{O})
+ (\widetilde{\rho}_0(n-1)|\hat{O})
\end{equation}

\nd Hence we have proved Eq. \eqref{QSDT1}.

(I): We have

\begin{equation}
\begin{split}
&tr(\hat{\rho}_i)=(\hat{\rho}_i|\hat{1})=\langle symb(\hat{\rho}_i), symb(\hat{1})\rangle=\langle \overline{1}_{A_i},1_{X}\rangle=\int_{\Gamma}d^{2(N+1)}\phi\ \overline{1}_{A_i}(\phi)=\int_{\Gamma}d^{2(N+1)}\phi\ [1/\mu(A_i)]1_{A_i}(\phi)\\
&=[1/\mu(A_i)]\int_{\Gamma}d^{2(N+1)}\phi\ 1_{A_i}(\phi)=[1/\mu(A_i)]\int_{A_i}d^{2(N+1)}\phi=[1/\mu(A_i)]\mu(A_i)=1
\end{split}
\end{equation}

\nd where we have used definition of $\overline{1}_{A_i}$
given by Eq. \eqref{characteristic1}, the Weyl
symbol property (see Eq. \eqref{WEYL}) and that $symb(\hat{1})=1_X$. Thus
$tr(\hat{\rho}_i)=1$, i.e. $\hat{\rho}_i \in \mathcal{N}$ for
all $i$. On the other hand if we apply Eq. \eqref{moyal2}
to $f=\overline{1}_{A_i}$ and
$g=\overline{1}_{A_j}$ we have

\begin{equation}\label{symb4}
\begin{split}
& symb(\hat{\rho}_i\hat{\rho}_j)=\overline{1}_{A_i}(\phi)\overline{1}_{A_j}(\phi)+0(\hbar)=\\
&=0(\hbar) \,\,\,\ if \,\ i\neq j\\
&=\overline{1}_{A_i}(\phi)+0(\hbar)\,\,\,\ if \,\ i=j
\end{split}
\end{equation}

\nd Now applying $symb^{-1}$ to both sides of Eq. \eqref{symb4} we obtain

\begin{equation}
\begin{split}
& \hat{\rho}_i\hat{\rho}_j=0(\hbar) \,\,\,\ if \,\ i\neq j\\
&=\hat{\rho}_i+0(\hbar)\,\,\,\ if \,\ i=j
\end{split}
\end{equation}
On the other hand from Eqns.
\eqref{characteristic1} and \eqref{characteristic1bis} we have
$1_X=\sum_i \mu(A_i)\overline{1}_{A_i}$ and therefore
$symb^{-1}(1_X)=symb^{-1}(\sum_i \mu(A_i)\overline{1}_{A_i})=\sum_i
\mu(A_i)symb^{-1}(\overline{1}_{A_i})$ where we have used that
$symb^{-1}$ is a linear map.
Now if we call $\alpha_i=\mu(A_i)$
since $symb^{-1}(1_X)=\hat{1}$ and
$\hat{\rho}_i=symb^{-1}(\overline{1}_{A_i})$ then we obtain Eq. \eqref{QSDT3}.

(II): Due to part (II) of Theorem \ref{SDT2} and taking into account that we are working under hypothesis
of Theorem \ref{SDT} we have

\begin{equation}\label{characteristic2}
P\overline{1}_{A_i}=\overline{1}_{A_\alpha(i)}
\end{equation}

\nd where $\alpha:\{1,...,r\}\longrightarrow \{1,...,r\}$ is a
permutation which satisfies $\alpha (i) \neq \alpha (j)$ for $i \neq
j$ and thus operator $P$ permutes the functions
$\overline{1}_{A_i}$. Let $\hat{O}$ be an observable and
$g=symb(\hat{O})$. Then from Eq. \eqref{characteristic2} we have

\begin{equation}\label{characteristic3}
\langle P\overline{1}_{A_i},g
\rangle=\langle\overline{1}_{A_\alpha(i)},g\rangle
\end{equation}

\nd and noting that

\begin{equation}
\begin{split}
& \langle P\overline{1}_{A_i},g \rangle=\langle \overline{1}_{A_i},Ug \rangle=\langle symb(\hat{\rho}_i),symb(\hat{O}(1)) \rangle=(\hat{\rho}_i|\hat{O}(1))=(\hat{U}\hat{\rho}_i\hat{U}^{\dag}|\hat{O})\\
& \langle\overline{1}_{A_\alpha(i)},g\rangle=\langle
symb(\hat{\rho}_{\alpha(i)}),symb(\hat{O})
\rangle=(\hat{\rho}_{\alpha(i)}|\hat{O})
\end{split}
\end{equation}

\nd then from Eq. \eqref{characteristic2} we have that
$(\hat{U}\rho_i\hat{U}^{\dag}|\hat{O})=(\hat{\rho}_{\alpha(i)}|\hat{O})$.

(III): Let $\hat{O}$ be an observable and $\varepsilon > 0$. Then by condition (III) of Theorem \ref{SDT2}
we have

\begin{equation}
\|P^{n-1}Qf\|=\|Q_nf\|< \frac{\varepsilon}{max\{|O(\phi)|:\phi \in
\Gamma\}}=\frac{\varepsilon}{\|O\|_{\infty}}
\,\,\,\,\,\,\,\,\,\,\ with \,\,\ O=symb(\hat{O})
\end{equation}

\nd Then

\begin{equation}\label{symb5}
(\widetilde{\rho}_0(n-1)|\hat{O})=\langle
symb(\widetilde{\rho}_0(n-1)), symb(\hat{O})\rangle=\langle Q_n f,
O\rangle\leq \|Q_nf\|\|O\|_{\infty}<\varepsilon
\end{equation}

\nd Therefore, from Eq. \eqref{symb5} it follows that
$(\widetilde{\rho}_0(n-1)|\hat{O}) \longrightarrow 0$.

\end{proof}

\subsection{Quantum Ergodic Hierarchy Levels: Ergodic, Mixing and Exact}
Theorems 7 and 8 can be used simultaneously to determine ergodicity, mixing or exactness just by looking at the terms in the sums on the right hand side of Eq. \eqref{SDT1 decomposition} or Eq. \eqref{SDT2 decomposition}. This observation holds also for the Quantum
Spectral Decomposition Theorem (QSDT, Theorem 9) simply because this is a translation to
the quantum language of its original version for dynamical systems. Therefore,
we can apply theorem 8 and QSDT to obtain the following theorem
in order to characterize the Quantum
Ergodic Hierarchy levels (see \cite{0} theorems 1,2 pags. 261, 263).

\begin{theorem}(\emph{Quantum Ergodic Hierarchy levels: Ergodic, Exact and Mixing})\label{QSDT AND HE}
Let $S$ be a quantum system. Let $\hat{\rho} \in \mathcal{N}$ and
let $\hat{O}$ be an observable. If $\hat{U}$ is a time evolution operator (i.e. representing a discretized time evolution in the sense of remark $(I)$ of section 4) whose classical analogue $T$ (i.e. representing the time evolution of the analogue classical system in the sense of remark $(II)$ of section 4) has a constrictive Markov operator $P$. Then

\begin{enumerate}

\item[$(I)$] $\hat{U}$ is \emph{ergodic} $\Longleftrightarrow$ permutation $\{\alpha(1),...,\alpha(r)\}$ of sequence $\{1,...,r\}$ is \emph{cyclical} (that is, for which there no is invariant subset).

\item[$(II)$] If $r=1$ in representation of Eq. \eqref{QSDT1} $\Longrightarrow$ $\hat{U}$ is \emph{exact}.

\item[$(III)$] If $\hat{U}$ is \emph{mixing} \emph{(see \cite{0} theorem 1 pag. 261)} $\Longrightarrow$ $r=1$ in representation of Eq. \eqref{QSDT1}.

\end{enumerate}

\nd Again, as in theorem 8 we have $\hat{U}$ exact $\Longleftrightarrow$ $\hat{U}$ mixing $\Longleftrightarrow$ $r=1$.

\end{theorem}

\begin{proof}
It is enough to use Theorems 8 and 9 simultaneously.
\end{proof}

\nd In following subsections we examine the QEH levels
established by Theorem 10 and its consequences in more detail. Since Theorem
10 does not distinguish between mixing and exact we will only discuss the mixing case (see remark 5.5.1 of \cite{LM}).

\subsubsection{A consequence of QSDT: Homogenization of the mixing level}

Consider a quantum system that is \emph{mixing}. Then from Theorem 10 and Eq.
\eqref{QSDT1} it follows that

\begin{equation}\label{QSDT consequence mixing}
(\hat{\rho}(n)|\hat{O})=(\hat{\rho}|\hat{O}_1)(\hat{\rho}_1|\hat{O})
+ (\widetilde{\rho}_0(n-1)|\hat{O})
\end{equation}

\nd where $\hat{\rho}_1$ is a pure state (in the classical limit) and
$\hat{O}_1$ is an observable which does not depend on the observable
$\hat{O}$. Further, since quantum system is mixing then it has a weak
limit $\hat{\rho}_{\ast}$ such that $lim_{n\rightarrow
\infty}(\hat{\rho}(n)|\hat{O})=(\hat{\rho}_{\ast}|\hat{O})$. From
this limit and Eq. \eqref{QSDT consequence mixing} we
have

\begin{equation}\label{QSDT consequence mixing2}
(\hat{\rho}_{\ast}|\hat{O})=lim_{n\rightarrow
\infty}(\hat{\rho}|\hat{O}_1)(\hat{\rho}_1|\hat{O}) +
lim_{n\rightarrow
\infty}(\widetilde{\rho}_0(n-1)|\hat{O})=(\hat{\rho}|\hat{O}_1)(\hat{\rho}_1|\hat{O})
\end{equation}

\nd since $lim_{n\rightarrow \infty}(\widetilde{\rho}_0(n-1)|\hat{O})=0$
(see (III) of Theorem 9). Now, if we make $\hat{O}=\hat{1}$ in Eq. \eqref{QSDT consequence mixing2} given that
$(\hat{\rho}_{\ast}|\hat{1})=tr(\hat{\rho}_{\ast})=1$ and
$(\hat{\rho}_1|\hat{1})=tr(\hat{\rho}_1)=1$ we obtain
$(\hat{\rho}|\hat{O}_1)=1$ then
$(\hat{\rho}_{\ast}|\hat{O})=(\hat{\rho}_1|\hat{O})$ for all
$\hat{O}$ observable. Therefore, $\hat{\rho}_1=\hat{\rho}_{\ast}$.
In other words, mixing level is physically responsible for the \emph{homogenization} of $\hat{\rho}$ and its
evolution towards the weak limit $\hat{\rho}_1$ which is a pure state. In this sense we can say that
QSDT gives a physical interpretation of the mixing level.

\subsubsection{The ergodic level: Oscillation plus a term going to zero}

From Theorem 10 we can obtain a necessary and sufficient condition
for ergodicity. A quantum system is \emph{ergodic} if and only if
permutation $\alpha$ of

\begin{equation}\label{QSDT consequence ergodic1}
(\hat{\rho}(n)|\hat{O})=\sum_{i=1}^{r}(\hat{\rho}|\hat{O}_{\alpha^{-n}(i)})(\hat{\rho}_i|\hat{O})
+ (\widetilde{\rho}_0(n-1)|\hat{O})
\end{equation}
is cyclical and $(\widetilde{\rho}_0(n-1)|\hat{O})\rightarrow 0$.
Since $\alpha$ is cyclical there is an integer $N>0$ such that
$\alpha^{N}(i)=i$ and $\alpha^{-N}(i)=i$ for $i=1,...,r$. That is,
the inverse permutation $\alpha^{-1}$ operates on the indices
$i=1,...,r$ as

\begin{equation}\label{QSDT consequence ergodic2}
(1,...,r)\longrightarrow^{\alpha^{-1}}(2,3,...,r-2,r-1,r,1)\longrightarrow^{\alpha^{-2}}(3,4,...,r-1,r,1,2)...\longrightarrow^{\alpha^{-k}}...\longrightarrow^{\alpha^{-N}}(1,...,r)
\end{equation}
After $N$ successive steps we are back to the original cycle
$(1,...,r)$. This behavior indicates that sum of Eq. \eqref{QSDT
consequence ergodic1} will also return to its original value after
$N$ successive time instants. Then sum of Eq. \eqref{QSDT
consequence ergodic1} is periodic with a period equal to $N$, i.e.
with the same period as cycle ${\alpha^{-1}}$. More precisely, we have

\begin{equation}\label{QSDT consequence ergodic3}
\begin{split}
&(\hat{\rho}(0)|\hat{O})=(\hat{\rho}|\hat{O}_{1})(\hat{\rho}_1|\hat{O})+(\hat{\rho}|\hat{O}_{2})(\hat{\rho}_2|\hat{O})+...+(\hat{\rho}|\hat{O}_{r-1})(\hat{\rho}_{r-1}|\hat{O})+(\hat{\rho}|\hat{O}_r)(\hat{\rho}_{r}|\hat{O})+(\widetilde{\rho}_0(-1)|\hat{O}) \\
&(\hat{\rho}(1)|\hat{O})=(\hat{\rho}|\hat{O}_{2})(\hat{\rho}_1|\hat{O})+(\hat{\rho}|\hat{O}_{3})(\hat{\rho}_2|\hat{O})+...+(\hat{\rho}|\hat{O}_{r})(\hat{\rho}_{r-1}|\hat{O})+(\hat{\rho}|\hat{O}_1)(\hat{\rho}_{r}|\hat{O})+(\widetilde{\rho}_0(0)|\hat{O}) \\
&(\hat{\rho}(2)|\hat{O})=(\hat{\rho}|\hat{O}_{3})(\hat{\rho}_1|\hat{O})+(\hat{\rho}|\hat{O}_{4})(\hat{\rho}_2|\hat{O})+...+(\hat{\rho}|\hat{O}_{1})(\hat{\rho}_{r-1}|\hat{O})+(\hat{\rho}|\hat{O}_2)(\hat{\rho}_{r}|\hat{O})+(\widetilde{\rho}_0(1)|\hat{O}) \\
&.\\
&.\\
&.\\
&(\hat{\rho}(N)|\hat{O})=(\hat{\rho}|\hat{O}_{1})(\hat{\rho}_1|\hat{O})+(\hat{\rho}|\hat{O}_{2})(\hat{\rho}_2|\hat{O})+...+(\hat{\rho}|\hat{O}_{r-1})(\hat{\rho}_{r-1}|\hat{O})+(\hat{\rho}|\hat{O}_r)(\hat{\rho}_{r}|\hat{O})+(\widetilde{\rho}_0(N-1)|\hat{O}) \\
&(\hat{\rho}(N+1)|\hat{O})=(\hat{\rho}|\hat{O}_{2})(\hat{\rho}_1|\hat{O})+(\hat{\rho}|\hat{O}_{3})(\hat{\rho}_2|\hat{O})+...+(\hat{\rho}|\hat{O}_{r})(\hat{\rho}_{r-1}|\hat{O})+(\hat{\rho}|\hat{O}_1)(\hat{\rho}_{r}|\hat{O})+(\widetilde{\rho}_0(N)|\hat{O}) \\
\end{split}
\end{equation}
Then
$\sum_{i=1}^{r}(\hat{\rho}|\hat{O}_{\alpha^{-(N+1)}(i)})(\hat{\rho}_i|\hat{O})=\sum_{i=1}^{r}(\hat{\rho}|\hat{O}_{\alpha^{-1}(i)})(\hat{\rho}_i|\hat{O})$.
That is, if we call $F(n)=$

$\sum_{i=1}^{r}(\hat{\rho}|\hat{O}_{\alpha^{-n}(i)})(\hat{\rho}_i|\hat{O})$
then
\begin{equation}
F(N+1)=F(1)
\end{equation}
Thus sum $F(n)$ is periodic with a period equal to
$N$. Moreover, permutation $\alpha$ is cyclical then $N=r$. Therefore, we see that in the ergodic case the mean values
are composed by an oscillating part plus a term which tends to zero for
large times ($n$ goes to $\infty$). From this remark we can obtain a sufficient and necessary condition for ergodicity.
Let $S$ be a quantum system having a constrictive markovian Frobenius-Perron operator $P$ associated with the temporal evolution $T$ of its classical analogue. Then we have

\begin{itemize}
\item[$(\bullet)$] $S$ is ergodic $\Longleftrightarrow$ for all discretized time evolution\footnote{Again, here we are considering that $\hat{U}(n)$ is the one given by the Hamiltonian or at most given by the Floquet operator. That is, $\hat{U}(n)=e^{-i\frac{\hat{H}}{\hbar}\alpha n}$ (Hamiltonian evolution) or
    $\hat{U}(n)=[\hat{U}(\tau)]^n$ (Floquet evolution, see Eq. 4.1.14 of \cite{stockmann}).} (with discrete time steps at constants intervals) any mean value $(\hat{\rho}(n)|\hat{O})$ has only two terms, one of which is an oscillatory function of $n$ (with a period $N$ equals to the number of terms $r$ of the QSDT decomposition given by the Eq. \eqref{QSDT1}) and the other goes to zero for $n\rightarrow\infty$.
\end{itemize}

\noindent In next section the physical interpretation of condition $(\bullet)$
is studied which is the key to the characterization of the ergodic
systems spectrum.

\section{QSDT and the quantum spectrum}

In addition to the characterization of mean values form of quantum ergodic hierarchy
QSDT provides also a link between QEH levels and spectrum. In this section we extend this
relationship to discrete case and continuous case, and when both types of spectrum are present.
We point out the formalism we use was introduced by Antoniou et al in order to a give a rigorous mathematical description on the algebraic formalism of quantum
mechanics \cite{antoniou}.

Moreover, although condition $(\bullet)$ of previous section is a general result it is of great interest to examine the particular case  $\hat{U}(n)=e^{-i\frac{\hat{H}}{\hbar}\alpha n}$ (i.e. the evolution operator is given by the Hamiltonian) where $\alpha$ is a real parameter which defines the discrete time steps. In this section we study what conditions for ergodicity and mixing can be obtained under this assumption.

\subsection{Discrete spectrum}

Let $\hat{\rho}\in\mathcal{N}$ be a state and let $\hat{O}$ be an
observable. For simplicity, we assume the spectrum is finite and discrete. Also, we assume real eigenvalues. Let $E_1,
E_2,..., E_N$ be the energies of the system with their corresponding frequencies
$\omega_1=\frac{E_1}{\hbar}, \omega_2=\frac{E_2}{\hbar},...,\omega_N=\frac{E_N}{\hbar}$. Let $|1\rangle, |2\rangle,...,|N\rangle$ be their corresponding eigenvectors. We express any initial state $\hat{\rho}$ in the basis $\{|i\rangle\}_{i=1}^{N}$ as
\begin{equation}\label{DISCRETE SPECTRUM10}
\hat{\rho}=\sum_{i=1}^{N}\rho_i|i\rangle\langle i|+\sum_{j\neq j^{\prime}}^{N}\rho_{j j^{\prime}}|j\rangle\langle j^{\prime}|
\end{equation}
\noindent Then considering $\hat{U}(n)=e^{-i\frac{\hat{H}}{\hbar}\alpha n}$ the mean value of
$\hat{O}$ in $\hat{\rho}$ after $n$ time steps is

\begin{equation}\label{DISCRETE SPECTRUM1}
(\hat{\rho}(n)|\hat{O})=\sum_{i=1}^{N}\rho_i O_i+\sum_{j\neq j^{\prime}}^{N}\rho_{j j^{\prime}}^{\ast}O_{j j^{\prime}}e^{-i(\omega_j-\omega_j^{\prime})\alpha n}
\end{equation}

\noindent On the other hand, by
QSDT we have

\begin{equation}\label{DISCRETE SPECTRUM2}
(\hat{\rho}(n)|\hat{O})=\sum_{i=1}^{r}\lambda_{\alpha^{-n}(i)}(\hat{\rho}_i|\hat{O})
+ (\widetilde{\rho}_0(n-1)|\hat{O})
\end{equation}

\noindent Then, combining Eqns. \eqref{DISCRETE SPECTRUM1} and \eqref{DISCRETE SPECTRUM2} we obtain

\begin{equation}\label{DISCRETE SPECTRUM3}
\sum_{i=1}^{N}\rho_i O_i+\sum_{j\neq j^{\prime}}^{N}\rho_{j j^{\prime}}^{\ast}O_{j j^{\prime}}e^{-i(\omega_j-\omega_j^{\prime})\alpha n}=\sum_{i=1}^{r}\lambda_{\alpha^{-n}(i)}(\hat{\rho}_i|\hat{O})
+ (\widetilde{\rho}_0(n-1)|\hat{O})
\end{equation}

\noindent From Eq. \eqref{DISCRETE SPECTRUM3} we can make the following considerations.

\noindent (A): If we suppose the system is ergodic then $\sum_{i=1}^{r}\lambda_{\alpha^{-n}(i)}(\hat{\rho}_i|\hat{O})$
is a periodic function and because $(\widetilde{\rho}_0(n-1)|\hat{O})\rightarrow 0$ for $n\rightarrow
\infty$ it follows that $\sum_{i=1}^{r}\lambda_{\alpha^{-n}(i)}(\hat{\rho}_i|\hat{O})
+ (\widetilde{\rho}_0(n-1)|\hat{O})$ is a quasi-periodic function of $n$\footnote{If we add a function $g(n)$ going to zero for $n\rightarrow \infty$ to a periodic function $f(n)$ then the sum $f(n)+g(n)$ is a quasi-periodic function.}. The left hand side of Eq. \eqref{DISCRETE SPECTRUM3} contains the term $\sum_{i=1}^{N}\rho_i O_i$ that is constant and therefore the sum $\sum_{j\neq j^{\prime}}^{N}\rho_{j j^{\prime}}^{\ast}O_{j j^{\prime}}e^{-i(\omega_j-\omega_j^{\prime})\alpha n}$ is a quasi-periodic function of $n$.

\noindent (B): Now suppose $\omega_1=\frac{E_1}{\hbar}, \omega_2=\frac{E_2}{\hbar},...,\omega_N=\frac{E_N}{\hbar}$ are rationally related. This means that there exists $k_{ij},l_{ij}\in\mathbb{N}$ such that
\begin{equation}\label{DISCRETE SPECTRUM RATIONAL}
\frac{\omega_i}{\omega_j}=\frac{k_{ij}}{l_{ij}} \,\,\,\,\,\,\,\,\, \forall i,j=1,...,N
\end{equation}
\noindent Now if we put $\omega_0=min\{\omega_1,...,\omega_N\}$ then we have\footnote{From Eq. \eqref{DISCRETE SPECTRUM RATIONAL} it follows that  $\frac{\omega_i}{\omega_0}=\frac{k_{i0}}{l_{i0}}$ for all $j=1,...,N$. Now if make $n_i=k_{i0}$ and $m_i=l_{i0}$ then we obtain the Eq. \eqref{DISCRETE SPECTRUM RATIONAL2}.}
\begin{equation}\label{DISCRETE SPECTRUM RATIONAL2}
\omega_j=\frac{n_j}{m_j}\omega_0 \,\,\,\,\,\,\,\,\, \forall j=1,...,N \,\ and \,\  n_{j},m_{j}\in\mathbb{N}
\end{equation}
\noindent On the other hand we know that each term\footnote{C.C denotes the complex conjugate.} $\rho_{ij}^{\ast}O_{ij}e^{-i(\omega_i-\omega_j)\alpha t}+C.C$ of non-diagonal part of $(\hat{\rho}(n)|\hat{O})$ (see Eq. \eqref{DISCRETE SPECTRUM1}) has a period $T_{ij}=\frac{2\pi}{\alpha(\omega_i-\omega_j)}$. The number $T_{ij}$ is not necessarily an integer and it depends on the values of $\omega_i,\omega_j,\alpha$. Moreover, using Eq. \eqref{DISCRETE SPECTRUM RATIONAL2} we can write $T_{ij}$ as

\begin{equation}\label{DISCRETE SPECTRUM RATIONAL3}
T_{ij}=\frac{2\pi \ m_i m_j}{\alpha(n_i-n_j)\omega_0} \,\,\,\,\,\,\,\,\, \forall i,j=1,...,N
\end{equation}
\noindent Now the trick is to choose a parameter $\alpha$ such that non-diagonal part of $(\hat{\rho}(n)|\hat{O})$ is periodic. Consider the lowest common multiple of products $m_i m_j$ and the greatest common divisor of all differences $n_i-n_j$ denoted by $LCM\{m_i m_j\}$ and $GCD\{n_im_j-n_jm_i\}$ respectively. If we choose $\alpha=\frac{2\pi}{GCD\{n_im_j-n_jm_i\}\omega_0}$ then it is not difficult to see that
$\sum_{i\neq j}^{N}\rho_{i j}^{\ast}O_{i j}e^{-i(\omega_i-\omega_j)\alpha n}$ is a periodic function of $n$. For this, it is enough to show that each term $\rho_{ij}^{\ast}O_{ij}e^{-i(\omega_i-\omega_j)\alpha n}+C.C$ takes the same value as in $n=0$ (i.e. $\rho_{ij}^{\ast}O_{ij}+C.C$) after $n$ steps where $n=LCM\{m_i m_j\}$. If $n=LCM\{m_i m_j\}$ and $\alpha=\frac{2\pi}{GCD\{n_im_j-n_jm_i\}\omega_0}$ we have

\begin{equation}\label{DISCRETE SPECTRUM RATIONAL4}
\begin{split}
&\rho_{ij}^{\ast}O_{ij}e^{-i(\omega_i-\omega_j)\alpha n}+C.C=\rho_{ij}^{\ast}O_{ij}e^{-i(\frac{n_i}{m_i}-\frac{n_j}{m_j})\omega_0\alpha n}+C.C=
\rho_{ij}^{\ast}O_{ij}e^{-i(\frac{n_im_j-n_jm_i}{m_i m_j})\omega_0\alpha n}+C.C=\\
&\rho_{ij}^{\ast}O_{ij}e^{-i(\frac{n_im_j-n_jm_i}{m_i m_j})\omega_0\frac{2\pi}{GCD\{n_im_j-n_jm_i\}\omega_0} LCM\{m_i m_j\}}+C.C=\\
&\rho_{ij}^{\ast}O_{ij}e^{-i2\pi(\frac{n_im_j-n_jm_i}{GCD\{n_im_j-n_jm_i\}}\frac{LCM\{m_i m_j\}}{m_i m_j})}+C.C
\end{split}
\end{equation}
\noindent Since $n_im_j-n_jm_i$ is divisible by $GCD\{n_im_j-n_jm_i\}$ and $LCM\{m_i m_j\}$ is divisible by $m_i m_j$ then there exists $K_{ij},L_{ij}\in\mathbb{Z}$ such that
\begin{equation}\label{DISCRETE SPECTRUM RATIONAL5}
\begin{split}
&K_{ij}=\frac{n_im_j-n_jm_i}{GCD\{n_im_j-n_jm_i\}}\\
&L_{ij}=\frac{LCM\{m_i m_j\}}{m_i m_j}
\end{split}
\end{equation}
\noindent Thus combining Eqns. \eqref{DISCRETE SPECTRUM RATIONAL4} and \eqref{DISCRETE SPECTRUM RATIONAL5} we have
\begin{equation}\label{DISCRETE SPECTRUM RATIONAL6}
\begin{split}
&\rho_{ij}^{\ast}O_{ij}e^{-i(\omega_i-\omega_j)\alpha n}+C.C=\rho_{ij}^{\ast}O_{ij}e^{-i2\pi K_{ij}L_{ij}}+C.C=
\rho_{ij}^{\ast}O_{ij}+C.C\\
\end{split}
\end{equation}

\noindent Therefore $\sum_{j\neq j^{\prime}}^{N}\rho_{j j^{\prime}}^{\ast}O_{j j^{\prime}}e^{-i(\omega_j-\omega_j^{\prime})n}$ is a periodic function of $n$. Now if we assume the system is ergodic then $\sum_{i=1}^{r}\lambda_{\alpha^{-n}(i)}(\hat{\rho}_i|\hat{O})$ is a periodic function of $n$ (see section 4.2.2.). Then we have

\begin{equation}\label{DISCRETE SPECTRUM4}
\sum_{i=1}^{N}\rho_i O_i+\sum_{j\neq j^{\prime}}^{N}\rho_{j j^{\prime}}^{\ast}O_{j j^{\prime}}e^{-i(\omega_j-\omega_j^{\prime})n}-\sum_{i=1}^{r}\lambda_{\alpha^{-n}(i)}(\hat{\rho}_i|\hat{O})=(\widetilde{\rho}_0(n-1)|\hat{O})
\end{equation}
is a non-trivial periodic function of $n$ going to zero for $n\longrightarrow \infty$. Contradiction. Then $S$ is not ergodic as we assumed.

Summing up,
given a quantum system $S$ of finite discrete spectrum with a constrictive and markovian Frobenius-Perron operator $P$ associated with the temporal evolution of its classical analogue and considering $\hat{U}(n)=e^{-i\frac{\hat{H}}{\hbar}\alpha n}$, from the considerations (A) and (B) we obtain the following conditions for ergodicity.

\begin{itemize}
\item[$(\star)$] $S$ is ergodic $\Longrightarrow$ quantum mean values $(\hat{\rho}(n)|\hat{O})$ are quasi-periodic functions of $n$ for all state $\hat{\rho}$ and observable $\hat{O}$.
\item[$(\star\star)$] If frequencies $\omega_1=\frac{E_1}{\hbar}, \omega_2=\frac{E_2}{\hbar},...,\omega_N=\frac{E_N}{\hbar}$ are rationally related $\Longrightarrow$ $S$ is not ergodic.

\end{itemize}

\noindent It is interesting to note that if we associate frequencies $\omega_1=\frac{E_1}{\hbar}, \omega_2=\frac{E_2}{\hbar},...,\omega_N=\frac{E_N}{\hbar}$ with angular velocities $\omega_1, \omega_2,...,\omega_N$ of $N$ independent and autonomous oscillators then
the condition $(\star\star)$ is the same that the one obtained for
\emph{rotation on the torus} (see pag. 190-193 of \cite{LM}). A possible explanation for this fact is as follows. QSDT allows to express mean values $(\hat{\rho}(n)|\hat{O})$ of a quantum system $S$ such that classical properties like ergodicity, mixing can be translated into quantum language in the same manner as the Spectral Decomposition Theorem (SDT), i.e. looking the number of terms and the periodicity of SDT (QSDT) decomposition (see Eqns. \eqref{SDT2 decomposition} and \eqref{QSDT1}). Put in other words, as well as SDT only gives conditions for classical ergodicity (mixing) and does not distinguish between two ergodic (mixing) dynamical systems the same happens with QSDT. If we have two quantum systems $S_1,S_2$ of finite discrete (and real) spectrum QSDT does not distinguish whether any of them is an harmonic oscillator, a particle in a box, etc. But it only says that $S_1$ or $S_2$ are not ergodic if any of them satisfy that $\omega_1=\frac{E_1}{\hbar}, \omega_2=\frac{E_2}{\hbar},...,\omega_N=\frac{E_N}{\hbar}$ are rationally related. Equivalently, condition $(\star\star)$ can be read as

\begin{itemize}
\item[$(\star\star\star)$] $S$ is ergodic $\Longrightarrow$ frequencies $\omega_1=\frac{E_1}{\hbar}, \omega_2=\frac{E_2}{\hbar},...,\omega_N=\frac{E_N}{\hbar}$ are not rationally related.
\end{itemize}
\noindent In the next section we examine the case of discrete complex eigenvalues.

\subsubsection{Discrete complex eigenvalues}
Now we consider that eigenvalues are complex and finite. This is the case of quantum systems described by an effective non-hermitian Hamiltonian $H_{eff}$ that appears in atomic, molecular, nuclear physics and in chemical reactions. Let $\hat{\rho}\in\mathcal{N}$ be a sate and let $\hat{O}$ be an observable. As in previous section, we assume that any state $\hat{\rho}$ after $n$ successive steps is given by $\hat{U}(n)\hat{\rho}\hat{U}^{\dag}(n)$ where $U(n)=e^{-i\frac{H_{eff}}{\hbar}\alpha n}$ and $\alpha\in\mathbb{R}$ defines the time steps. Let $E_1=\omega_1+i\gamma_1,
E_2=\omega_2+i\gamma_2,...,E_N=\omega_N+i\gamma_N$ be the complex eigenvalues of $H_{eff}$. Non-Hermiticity of $H_{eff}$ yields two set of eigenvectors called $\langle \widetilde{1}|,\langle \widetilde{2}|,...,\langle \widetilde{N}|$ left eigenvectors and $|1\rangle, |2\rangle,...,|N\rangle$ right eigenvectors (see Eq. (2) of \cite{kuhl})

\begin{equation}\label{complex discrete1}
H_{eff}|j\rangle=E_j |j\rangle \,\,\,\,\,,\,\,\,\,\ \langle \widetilde{j}|H_{eff}=\langle \widetilde{j}|E_j \,\,\,\,\,\,\,\,\,\,   j=1,...,N
\end{equation}
satisfying conditions
\begin{equation}\label{complex discrete2}
\begin{split}
&\langle \widetilde{j}|k\rangle=\delta_{jk} \,\,\,\,\,\,\,\,\,\,   j,k=1,...,N \,\,\,\,\,\,\,\,\,\, (bi-orthogonality)\\
&\sum_{j=1}^N |j\rangle \langle\widetilde{j}|=\mathbb{I} \,\,\,\,\,\,\,\,\,\, (completeness)
\end{split}
\end{equation}

\nd where $E_j=\omega_j+i\gamma_j$ gives the energy $\omega_j$ and the resonance width $-\gamma_j>0$ of $jth$ resonance (see pag. 3 of \cite{kuhl}). In a more realistic case (see section 6) we can suppose that only first $K$ eigenvalues are real with $1\leq K\leq N$, i.e. $\gamma_j=0$ for all $j=1,...,K$. Again, we express $\hat{\rho}$ in the basis (co--basis) $\{|j\rangle\}_{j=1}^N$ ($\{|\langle\widetilde{k}|\}_{k=1}^N$)

\begin{equation}\label{complex discrete3}
\hat{\rho}=\sum_{i=1}^{N}\rho_i|i\rangle\langle \widetilde{i}|+\sum_{j\neq j^{\prime}}^{N}\rho_{j j^{\prime}}|j\rangle\langle \widetilde{j^{\prime}}|
\end{equation}
\nd From Eq. \eqref{complex discrete3} we have

\begin{equation}\label{complex discrete4}
\begin{split}
&\hat{\rho}(n)=\hat{U}(n)\hat{\rho}\hat{U}^{\dag}(n)=\hat{U}(n)\left(\sum_{i=1}^{N}\rho_i|i\rangle\langle \widetilde{i}|+\sum_{j\neq j^{\prime}}^{N}\rho_{j j^{\prime}}|j\rangle\langle \widetilde{j^{\prime}}|\right)\hat{U}(n)^{\dag}=\\
&\sum_{i=1}^{K}\rho_i|i\rangle\langle \widetilde{i}|+\sum_{j\neq j^{\prime}}^{K}\rho_{j j^{\prime}}^{\ast}e^{-i\frac{(\omega_j-\omega_j^{\prime})}{\hbar}\alpha n}|j\rangle\langle \widetilde{j^{\prime}}|+\sum_{j=K+1}^{N}\rho_je^{2\alpha n\frac{\gamma_j}{\hbar}}|j\rangle\langle \widetilde{j}|+\\
&\sum_{j\neq j^{\prime},j \,\ or j^{\prime} \in\{K+1,...,N\}}\rho_{j j^{\prime}}^{\ast}e^{-i\frac{(\omega_j-\omega_j^{\prime})}{\hbar}\alpha n}e^{\frac{(\gamma_j+\gamma_j^{\prime})}{\hbar}\alpha n}|j\rangle\langle \widetilde{j^{\prime}}|
\end{split}
\end{equation}
\nd Coefficients of last two sums of Eq. \eqref{complex discrete4} contains probabilities of transitions to eigenstates with complex eigenvalues called ``quasi-stationary states" whose resonance width $-\gamma_j$ is related to the lifetime $\Gamma_j$ by $-\gamma_j\sim\frac{\hbar}{\Gamma_j}$, see pag. 559 of \cite{landau}. Physically, $|K+1\rangle, |K+2\rangle,...,|N\rangle$ represent states of an open quantum system in interaction with an environment where the exponentially decreasing coefficients of last two sums Eq. \eqref{complex discrete4} can be interpreted as probabilities of finding a particle in one of the states $|K+1\rangle, |K+2\rangle,...,|N\rangle$ ``inside the system".

Next step is to show how quasi-stationary states can be interpreted in terms of QSDT. From Eq. \eqref{complex discrete4} the mean value of an observable $\hat{O}$ after $n$ time steps reads
\begin{equation}\label{complex discrete5}
\begin{split}
&(\hat{\rho}(n)|\hat{O})=\sum_{i=1}^{K}\rho_iO_i+\sum_{j\neq j^{\prime}}^{K}\rho_{j j^{\prime}}^{\ast}O_{j j^{\prime}}e^{-i\frac{(\omega_j-\omega_j^{\prime})}{\hbar}\alpha n}+\sum_{j=K+1}^{N}\rho_jO_je^{2\alpha n\frac{\gamma_j}{\hbar}}+\\
&\sum_{j\neq j^{\prime},j \,\ or j^{\prime} \in\{K+1,...,N\}}\rho_{j j^{\prime}}^{\ast}O_{j j^{\prime}}e^{-i\frac{(\omega_j-\omega_j^{\prime})}{\hbar}\alpha n}e^{\frac{(\gamma_j+\gamma_j^{\prime})}{\hbar}\alpha n}
\end{split}
\end{equation}
\nd Therefore, by QSDT we have

\begin{equation}\label{complex discrete6}
\begin{split}
&\sum_{i=1}^{r}\lambda_{\alpha^{-n}(i)}(\hat{\rho}_i|\hat{O})
+ (\widetilde{\rho}_0(n-1)|\hat{O})=\sum_{i=1}^{K}\rho_iO_i+\sum_{j\neq j^{\prime}}^{K}\rho_{j j^{\prime}}^{\ast}O_{j j^{\prime}}e^{-i\frac{(\omega_j-\omega_j^{\prime})}{\hbar}\alpha n}+\\
&\sum_{j=K+1}^{N}\rho_jO_je^{2\alpha n\frac{\gamma_j}{\hbar}}+\sum_{j\neq j^{\prime},j \,\ or j^{\prime} \in\{K+1,...,N\}}\rho_{j j^{\prime}}^{\ast}O_{j j^{\prime}}e^{-i\frac{(\omega_j-\omega_j^{\prime})}{\hbar}\alpha n}e^{\frac{(\gamma_j+\gamma_j^{\prime})}{\hbar}\alpha n}
\end{split}
\end{equation}
\nd Since $(\widetilde{\rho}_0(n-1)|\hat{O})\rightarrow0$ for $n\rightarrow\infty$ and last two sums on the right hand of Eq. \eqref{complex discrete6} are the only terms that decay then we have
\begin{equation}\label{complex discrete7}
\begin{split}
&\sum_{i=1}^{r}\lambda_{\alpha^{-n}(i)}(\hat{\rho}_i|\hat{O})=\sum_{i=1}^{K}\rho_iO_i+\sum_{j\neq j^{\prime}}^{K}\rho_{j j^{\prime}}^{\ast}O_{j j^{\prime}}e^{-i\frac{(\omega_j-\omega_j^{\prime})}{\hbar}\alpha n}\\
&(\widetilde{\rho}_0(n-1)|\hat{O})=\sum_{j=K+1}^{N}\rho_jO_je^{2\alpha n\frac{\gamma_j}{\hbar}}+\sum_{j\neq j^{\prime},j \,\ or j^{\prime} \in\{K+1,...,N\}}\rho_{j j^{\prime}}^{\ast}O_{j j^{\prime}}e^{-i\frac{(\omega_j-\omega_j^{\prime})}{\hbar}\alpha n}e^{\frac{(\gamma_j+ \gamma_j^{\prime})}{\hbar}\alpha n}
\end{split}
\end{equation}
\nd From Eq. \eqref{complex discrete7} and Theorem 10 we conclude that the number of ``quasi-stationary" states $|K+1\rangle, |K+2\rangle,...,|N\rangle$
determines chaotic nature of the quantum system. For instance, if we have only one real eigenvalue $E_1$ then we have maximum number of quasi-stationary states $|2\rangle, |3\rangle,...,|N\rangle$. In this case $K=1$ and  $\sum_{i=1}^{r}\lambda_{\alpha^{-n}(i)}(\hat{\rho}_i|\hat{O})$ is equal to $\rho_1O_1$. Then $r=1$, i.e. the system is mixing. When $K=2$ we have two real eigenvalues $E_1,E_2$ and $\sum_{i=1}^{r}\lambda_{\alpha^{-n}(i)}(\hat{\rho}_i|\hat{O})$ is equal to
$\rho_1O_1+\rho_2O_2+\rho_{1 2}^{\ast}O_{1 2}e^{-i\frac{(\omega_1-\omega_2)}{\hbar}\alpha n}+\rho_{21}^{\ast}O_{21}e^{-i\frac{(\omega_2-\omega_1)}{\hbar}\alpha n}$ that (by the arguments of the consideration (B) of the previous section) is an oscillatory function of $n$ if $\omega_1,\omega_2$ are rationally related and taking $\alpha=\frac{2\pi}{GCD\{n_1m_2-n_2m_1\}\omega_0}$, i.e. $\omega_1=\frac{n_1}{m_1}\omega_0$ and $\omega_2=\frac{n_2}{m_2}\omega_0$. Then by condition $(\bullet)$ of section 4.2.2 we have the system is ergodic.

We can see that in the finite spectrum case we have a substantial difference between having real eigenvalues and complex eigenvalues. In the case of real eigenvalues (see condition $(\star\star)$ of section 5.1) we have seen that natural frequencies rationally related imply the system is not ergodic. While in the complex case if the real part of eigenvalues are rationally related then from the first line of Eq. \eqref{complex discrete7} it follows the system is ergodic\footnote{I.e., using the same arguments of (B) of section 5.1.}. Moreover, in the complex case the decay term of QSDT, i.e. $(\widetilde{\rho}_0(n-1)|\hat{O})$, is exponentially decreasing and provided by resonances widths $\gamma_j$ as we can see from the Eq. \eqref{complex discrete7}.

Summing up, we see that QSDT gives a characterization of open quantum systems of discrete complex eigenvalues where exponentially decreasing terms of mean values (the right hand on second line of Eq. \eqref{complex discrete7}) are associated with the decay term of QSDT, $(\widetilde{\rho}_0(n-1)|\hat{O})$. All the analysis described in this section will be illustrated in some detail with examples in section 6.

\subsection{Continuous spectrum}

Now we assume that spectrum is continuous being $E_{\omega}=\hbar \omega\in
[0,\infty)$ the energies, $|\omega\rangle$ the
generalized eigenvectors and $\omega=\frac{E_{\omega}}{\hbar}$ the natural frequencies. Let $\hat{\rho}$ be a state and let $\hat{O}$ be an observable. In order to obtain an approach to equilibrium we restrict the space
of observables and consider only \emph{van Hove Observables} (see \cite{van hove1, van hove2} for a more detail discussion). There is no loss of generality in this restriction since observables not belonging to van Hove space are not experimentally accessible (see
\cite{MARIO OLIMPIA2} for a complete argument). The components of a
van Hove observable $\hat{O}_R$ are
$O_R(\omega,\omega^{\prime})=O(\omega)\delta(\omega-\omega^{\prime})+O(\omega,\omega^{\prime})$ where $O(\omega)\delta(\omega-\omega^{\prime}),O(\omega,\omega^{\prime})$ are \emph{singular} and \emph{regular}\footnote{By \emph{regular} we mean, among other properties, that $f(\omega,\omega^{\prime})^{\ast}O(\omega,\omega^{\prime})\in
L^{1}([0,\infty)\times[0,\infty))$ for all $f(\omega,\omega^{\prime})$.} part of $\hat{O}_R$ respectively .
Then we can expand $\hat{O}$ in the basis
$\{|\omega\rangle\langle\omega|,|\omega\rangle\langle\omega^{\prime}|\}$
as

\begin{equation}\label{CONTINUOUS SPECTRA1}
\hat{O}=\int_{0}^{\infty}O(\omega)|\omega\rangle\langle\omega|d\omega+
\int_{0}^{\infty}\int_{0}^{\infty}O(\omega,\omega^{\prime})|\omega\rangle\langle\omega^{\prime}|
\end{equation}
Therefore, mean value of $\hat{O}$ in $\hat{\rho}$ after $n$ time steps is\footnote{Again, $\hat{\rho}(n)=\hat{U}(n)\hat{\rho}\hat{U}^{\dag}(n)$ with $U(n)=e^{-i\frac{\hat{H}}{\hbar}\alpha n}$ and $\alpha\in\mathbb{R}$.}
\begin{equation}\label{CONTINUOUS SPECTRA3}
\begin{split}
&(\hat{\rho}(n)|\hat{O})=\int_{0}^{\infty}\rho(\omega)^{\ast}O(\omega)d\omega+
\int_{0}^{\infty}\int_{0}^{\infty}\rho(\omega,\omega^{\prime})^{\ast}O(\omega,\omega^{\prime})e^{-i(\omega-\omega^{\prime})\alpha n}d\omega d\omega^{\prime}\\
\end{split}
\end{equation}
On the other hand, by QSDT we have
\begin{equation}\label{CONTINUOUS SPECTRA4}
\begin{split}
&(\hat{\rho}(n)|\hat{O})=\sum_{i=1}^{r}(\hat{\rho}|\hat{O}_{\alpha^{-n}(i)})(\hat{\rho}_i|\hat{O}) + (\widetilde{\rho}_0(n-1)|\hat{O})=\\
&=\int_{0}^{\infty}\rho(\omega)^{\ast}O(\omega)d\omega+
\int_{0}^{\infty}\int_{0}^{\infty}\rho(\omega,\omega^{\prime})^{\ast}O(\omega,\omega^{\prime})e^{-i(\omega-\omega^{\prime})\alpha n}d\omega d\omega^{\prime}\\
\end{split}
\end{equation}
If we assume that
$\rho(\omega,\omega^{\prime})^{\ast}O(\omega,\omega^{\prime})\in
L^{1}([0,\infty)\times[0,\infty))$ then by Riemann-Lebesgue
Lemma we have
\begin{equation}\label{CONTINUOUS SPECTRA5}
\int_{0}^{\infty}\int_{0}^{\infty}\rho(\omega,\omega^{\prime})^{\ast}O(\omega,\omega^{\prime})e^{-i(\omega-\omega^{\prime})\alpha n}d\omega
d\omega^{\prime}\longrightarrow0
\end{equation}
when $n\longrightarrow\infty$. That is, using van Hove observables and assuming that
$\rho(\omega,\omega^{\prime})^{\ast}O(\omega,\omega^{\prime})\in
L^{1}([0,\infty)\times[0,\infty))$ then the system is \emph{mixing}.

Moreover, from Eqns. \eqref{CONTINUOUS SPECTRA3} and \eqref{CONTINUOUS
SPECTRA5} it follows that
$(\hat{\rho}(n)|\hat{O})\longrightarrow(\hat{\rho}_{\ast}|\hat{O})=\int_{0}^{\infty}\rho(\omega)^{\ast}O(\omega)d\omega$
as $n\longrightarrow\infty$ with
$\hat{\rho}_{\ast}=\int_{0}^{\infty}\rho(\omega)^{\ast}|\omega\rangle\langle\omega|d\omega$.

Therefore, $\hat{\rho}$ has weak limit $\hat{\rho}_{\ast}$.
Also, if system is mixing then by Theorem 10 we have $r=1$ in the sum of \eqref{CONTINUOUS SPECTRA4}. Then we have
\begin{equation}\label{CONTINUOUS SPECTRA6}
\begin{split}
&(\hat{\rho}(n)|\hat{O})=(\hat{\rho}|\hat{O}_1)(\hat{\rho}_1|\hat{O}) + (\widetilde{\rho}_0(n-1)|\hat{O})=\\
&=\int_{0}^{\infty}\rho(\omega)^{\ast}O(\omega)d\omega+
\int_{0}^{\infty}\int_{0}^{\infty}\rho(\omega,\omega^{\prime})^{\ast}O(\omega,\omega^{\prime})e^{-i(\omega-\omega^{\prime})\alpha n}d\omega d\omega^{\prime}\\
\end{split}
\end{equation}
Now since $(\hat{\rho}|\hat{O}_1)(\hat{\rho}_1|\hat{O})$ and
$\int_{0}^{\infty}\rho(\omega)^{\ast}O(\omega)d\omega$ are constants
and

$\int_{0}^{\infty}\int_{0}^{\infty}\rho(\omega,\omega^{\prime})^{\ast}O(\omega,\omega^{\prime})e^{-i(\omega-\omega^{\prime})\alpha n}d\omega
d\omega^{\prime}\longrightarrow0,(\widetilde{\rho}_0(n-1)|\hat{O})\longrightarrow0$
then by Eq. \eqref{CONTINUOUS SPECTRA6} we obtain
\begin{equation}\label{CONTINUOUS SPECTRA7}
\begin{split}
&(\hat{\rho}|\hat{O}_1)(\hat{\rho}_1|\hat{O})=\int_{0}^{\infty}\rho(\omega)^{\ast}O(\omega)d\omega\\
&(\widetilde{\rho}_0(n-1)|\hat{O})=
\int_{0}^{\infty}\int_{0}^{\infty}\rho(\omega,\omega^{\prime})^{\ast}O(\omega,\omega^{\prime})e^{-i(\omega-\omega^{\prime})\alpha n}d\omega d\omega^{\prime}\\
\end{split}
\end{equation}
From Eq. \eqref{CONTINUOUS SPECTRA7} we see the physical interpretation
of term $(\widetilde{\rho}_0(n-1)|\hat{O})$: It
is the manifestation of Riemann-Lebesgue Lemma in closed quantum systems with continuous spectrum \cite{MARIO OLIMPIA} where observable space is the van Hove space. Therefore, the fact that mixing systems with continuous
spectrum are those with only one term
$(\hat{\rho}|\hat{O}_1)(\hat{\rho}_1|\hat{O})$ in decomposition given by Eq. \eqref{QSDT1}
is a consequence of QSDT where Riemann-Lebesgue Lemma is contained in
term $(\widetilde{\rho}_0(n-1)|\hat{O})$. From these arguments it follows the following property.

Summing up, given a quantum system $S$ of continuous spectrum with a constrictive markovian Frobenius-Perron operator $P$ associated with temporal evolution of its classical analogue and considering $U(n)=e^{-i\frac{\hat{H}}{\hbar}\alpha n}$ we have

\begin{itemize}
\item[$(\ast)$] $S$ is mixing and the sum of QSDT decomposition (see Eq. \eqref{QSDT1}) only contains the term $(\hat{\rho}|\hat{O}_1)(\hat{\rho}_1|\hat{O})$ which is the constant part of quantum mean value of Eq. \eqref{CONTINUOUS SPECTRA3}, i.e. $\int_{0}^{\infty}\rho(\omega)^{\ast}O(\omega)d\omega$.
The term $(\widetilde{\rho}_0(n-1)|\hat{O})$ is the manifestation of Riemann-Lebesgue Lemma.

\end{itemize}


\subsection{The General Case: Discrete and Continuous Spectrum}

If both types of spectrum are present
according to Eqns. \eqref{DISCRETE SPECTRUM1} and \eqref{CONTINUOUS
SPECTRA3} we have
\begin{equation}\label{GENERAL CASE1}
\begin{split}
&(\hat{\rho}(n)|\hat{O})=\sum_{i=1}^{N}\rho_i O_i+
\sum_{j\neq j^{\prime}}^{N}\rho_{j j^{\prime}}^{\ast}O_{j j^{\prime}}e^{-i(\omega_j-\omega_j^{\prime})\alpha n}+\\
&\int_{0}^{\infty}\rho(\omega)^{\ast}O(\omega)d\omega+
\int_{0}^{\infty}\int_{0}^{\infty}\rho(\omega,\omega^{\prime})^{\ast}O(\omega,\omega^{\prime})e^{-i(\omega-\omega^{\prime})\alpha n}d\omega
d\omega^{\prime}
\end{split}
\end{equation}
where first two sums and last two integrals represent the discrete and continuous contributions to spectrum respectively. By QSDT we have

\begin{equation}\label{GENERAL CASE2}
(\hat{\rho}(n)|\hat{O})=\sum_{i=1}^{r}(\hat{\rho}|\hat{O}_{\alpha^{-n}(i)})(\hat{\rho}_i|\hat{O})+(\widetilde{\rho}_0(n-1)|\hat{O})
\end{equation}
From Eqns. \eqref{GENERAL CASE1} and \eqref{GENERAL CASE2} we obtain
\begin{equation}\label{GENERAL CASE3}
\begin{split}
&\sum_{i=1}^{N}\rho_i O_i+\sum_{j\neq j^{\prime}}^{N}\rho_{j j^{\prime}}^{\ast}O_{j j^{\prime}}e^{-i(\omega_j-\omega_j^{\prime})\alpha n}+\\
&\int_{0}^{\infty}\rho(\omega)^{\ast}O(\omega)d\omega+
\int_{0}^{\infty}\int_{0}^{\infty}\rho(\omega,\omega^{\prime})^{\ast}O(\omega,\omega^{\prime})e^{-i(\omega-\omega^{\prime})\alpha n}d\omega
d\omega^{\prime}=\\
&\sum_{i=1}^{r}(\hat{\rho}|\hat{O}_{\alpha^{-n}(i)})(\hat{\rho}_i|\hat{O})+(\widetilde{\rho}_0(n-1)|\hat{O})
\end{split}
\end{equation}
From Eq. \eqref{GENERAL CASE3} we can analyze the different cases. Let $S$ be a quantum system having both types of spectrum as in Eq. \eqref{GENERAL CASE1}. Suppose $S$ has a constrictive markovian Frobenius-Perron operator $P$ associated with the temporal evolution of its classical analogue. One would be tempted to think that this case is simply the superposition of discrete and continuous cases each separately, but this is not so.

We begin with the ergodic case.
From Theorem 10 and section 4.2.2 we know that $S$ is ergodic if and only if sum on right hand of Eq. \eqref{GENERAL CASE3} is periodic and we see the only term on left hand of Eq. \eqref{GENERAL CASE3} that can be periodic is $\sum_{j\neq j^{\prime}}^{N}\rho_{j j^{\prime}}^{\ast}O_{j j^{\prime}}e^{-i(\omega_j-\omega_j^{\prime})\alpha n}$. Here we see a substantial difference from discrete case where we have $S$ is ergodic $\Longrightarrow$ $\sum_{j\neq j^{\prime}}^{N}\rho_{j j^{\prime}}^{\ast}O_{j j^{\prime}}e^{-i(\omega_j-\omega_j^{\prime})\alpha n}$ is quasi-periodic (see condition $(\star)$ of section 5.1).
Furthermore, the term $(\widetilde{\rho}_0(n-1)|\hat{O})$ which goes to zero should be associated with the only decay term on left hand of Eq. \eqref{GENERAL CASE3}, i.e. the ``Riemann-Lebesgue term" $\int_{0}^{\infty}\int_{0}^{\infty}\rho(\omega,\omega^{\prime})^{\ast}O(\omega,\omega^{\prime})e^{-i(\omega-\omega^{\prime})\alpha n}d\omega
d\omega^{\prime}$. Then we have

\begin{equation}\label{GENERAL CASE4}
\begin{split}
&\sum_{i=1}^{N}\rho_i O_i+\sum_{j\neq j^{\prime}}^{N}\rho_{j j^{\prime}}^{\ast}O_{j j^{\prime}}e^{-i(\omega_j-\omega_j^{\prime})\alpha n}+
\int_{0}^{\infty}\rho(\omega)^{\ast}O(\omega)d\omega=\sum_{i=1}^{r}(\hat{\rho}|\hat{O}_{\alpha^{-n}(i)})(\hat{\rho}_i|\hat{O})\\
&\int_{0}^{\infty}\int_{0}^{\infty}\rho(\omega,\omega^{\prime})^{\ast}O(\omega,\omega^{\prime})e^{-i(\omega-\omega^{\prime})\alpha n}d\omega
d\omega^{\prime}=(\widetilde{\rho}_0(n-1)|\hat{O})
\end{split}
\end{equation}
\nd Therefore we can see that in the ergodic case decay term is only provided by continuous part of spectrum as a consequence of QSDT. By contrast, in the mixing case this situation is very different due to non-degeneration condition of quantum chaos. Let us see why this is so.

In the mixing case we have that $S$ is mixing if and only $r=1$ in the sum on right hand of Eq. \eqref{GENERAL CASE3} and this happens if and only if double sum on left hand of Eq. \eqref{GENERAL CASE3} is constant. In this case there is no oscillatory term or quasiperiodic term and all frequencies $\omega_j$ are equal, e.g. $\omega_j=\omega_0$ for all $j=1,...,N$. Then energy levels are degenerated. But if we recall that non-degeneration is a necessary condition of quantum chaos then this situation should not be physically admissible. Thus the only possibility that $\omega_j$ are non-degenerated without any oscillatory or quasiperiodic term is if double sum $\sum_{j\neq j^{\prime}}^{N}\rho_{j j^{\prime}}^{\ast}O_{j j^{\prime}}e^{-i(\omega_j-\omega_j^{\prime})\alpha n}$ of Eq. \eqref{GENERAL CASE3} can be approximated by a double integral in order to apply Riemann-Lebesgue Lemma. In other words, discrete part must be
quasicontinuous, i.e. adjacent energy levels are very close. More precisely, considering discrete part is supported in an interval $[a,b]$ such that\footnote{Since $[0,\infty)$ is the continuous part of spectrum and we assume non-degeneration.} $[a,b] \cap [0,\infty)=\emptyset$ with $N$ equispaced frequencies
$\omega_j=a+\frac{(j-1)(b-a)}{N-1}$ and $j=1,...,N$ then sums on left hand of Eq. \eqref{GENERAL CASE3} must be replaced by

\begin{equation}\label{GENERAL CASE5}
\begin{split}
&\sum_{i=1}^{N}\rho_i O_i\longrightarrow \sum_{j=1}^{N}\rho(\omega_j)O(\omega_j)\Delta\omega_j
\,\,\,\,\,\,\,\,\,\,\,\,\,\,\,\,\,\,\Delta\omega_j=\frac{b-a}{N-1},\Delta\omega_{j^{\prime}}=\frac{b-a}{N-1}  \\
&\sum_{j\neq j^{\prime}}^{N}\rho_{j j^{\prime}}^{\ast}O_{j j^{\prime}}e^{-i(\omega_j-\omega_j^{\prime})\alpha n}\longrightarrow
\sum_{j,j^{\prime}=1}^{N}\sum_{j\neq j^{\prime}}^{N}\rho(\omega_j,\omega_j^{\prime})^{\ast}
O(\omega_j,\omega_j^{\prime})e^{-i(\omega_j-\omega_j^{\prime})\alpha n}\Delta\omega_j\Delta\omega_{j^{\prime}}\\
\end{split}
\end{equation}

\noindent where $\Delta\omega_j=\frac{b-a}{N-1}$ is the length of segment $[\omega_j,\omega_{j+1}]$ and $\Delta\omega_j\Delta\omega_{j^{\prime}}=(\frac{b-a}{N-1})^2$ is the volume of square $[\omega_j,\omega_{j+1}]\times [\omega_{j^{\prime}},\omega_{j^{\prime}+1}]$. Now since discrete part is quasicontinuous then frequencies $\omega_j$ are very close and we have $\Delta\omega_j=\frac{b-a}{N-1}\ll1$, i.e. we can take limit $N\rightarrow \infty$ in sums of \eqref{GENERAL CASE5}. In such case we can approximate sums of Eq. \eqref{GENERAL CASE5} by integrals and we have the replacements\footnote{It is understood that integrals $\int\int_{[a,b]\times[a,b]}$ and $\int\int_{[a,b]\times[a,b]\backslash\{(\omega,\omega):\omega\in[a,b]\}}$ of Eq. \eqref{GENERAL CASE6} are equal since $\{(\omega,\omega):\omega\in[a,b]\}$ is a zero measure set of plane $\{(\omega,\omega^{\prime})\}=\mathbb{R}_{\geq0}\times \mathbb{R}_{\geq0}$. }

\begin{equation}\label{GENERAL CASE6}
\begin{split}
&\omega_j\longrightarrow \omega \,\,\,\,\,\,\,,\,\,\,\,\,\,\ \Delta\omega_{j}\longrightarrow d\omega
\,\,\,\,\,\,\,,\,\,\,\,\,\,\ \Delta\omega_j\Delta\omega_{j^{\prime}}\longrightarrow d\omega d\omega^{\prime}\\
&\sum_{j=1}^{N}\rho(\omega_j)O(\omega_j)\Delta\omega_j\longrightarrow \int_{[a,b]}\rho(\omega)O(\omega)d\omega \\
&\sum_{j,j^{\prime}=1}^{N}\sum_{j\neq j^{\prime}}^{N}\rho(\omega_j,\omega_j^{\prime})^{\ast}
O(\omega_j,\omega_j^{\prime})e^{-i(\omega_j-\omega_j^{\prime})\alpha n}\Delta\omega_j\Delta\omega_{j^{\prime}}\longrightarrow\\
&\int\int_{[a,b]\times[a,b]\backslash\{(\omega,\omega):\omega\in[a,b]\}}\rho(\omega,\omega^{\prime})^{\ast}O(\omega,\omega^{\prime})e^{-i(\omega-\omega^{\prime})\alpha n}d\omega d\omega^{\prime}=\\
&\int\int_{[a,b]\times[a,b]}\rho(\omega,\omega^{\prime})^{\ast}O(\omega,\omega^{\prime})e^{-i(\omega-\omega^{\prime})\alpha n}d\omega d\omega^{\prime}\\
\end{split}
\end{equation}

\nd Therefore, for the quasicontinuous case Eq. \eqref{GENERAL CASE3} becomes

\begin{equation}\label{GENERAL CASE7}
\begin{split}
&\int_{[a,b]}\rho(\omega)O(\omega)d\omega+\int_{0}^{\infty}\rho(\omega)^{\ast}O(\omega)d\omega=(\hat{\rho}|\hat{O}_{1})(\hat{\rho}_1|\hat{O})\\
&\int\int_{[a,b]\times[a,b]}\rho(\omega,\omega^{\prime})^{\ast}O(\omega,\omega^{\prime})e^{-i(\omega-\omega^{\prime})\alpha n} d\omega d\omega^{\prime}+\\
&\int_{0}^{\infty}\int_{0}^{\infty}\rho(\omega,\omega^{\prime})^{\ast}O(\omega,\omega^{\prime})e^{-i(\omega-\omega^{\prime})\alpha n}d\omega
d\omega^{\prime}=(\widetilde{\rho}_0(n-1)|\hat{O})
\end{split}
\end{equation}
\nd That is,
\begin{equation}\label{GENERAL CASE7}
\begin{split}
&\int_{[a,b]\cup [0,\infty)}\rho(\omega)O(\omega)d\omega=(\hat{\rho}|\hat{O}_{1})(\hat{\rho}_1|\hat{O})\\
&\int\int_{[a,b]\times[a,b]\cup [0,\infty)\times[0,\infty)}\rho(\omega,\omega^{\prime})^{\ast}O(\omega,\omega^{\prime})e^{-i(\omega-\omega^{\prime})\alpha n} d\omega d\omega^{\prime}=(\widetilde{\rho}_0(n-1)|\hat{O})
\end{split}
\end{equation}
\nd Now if we apply Riemann-Lebesgue to double integral of Eq. \eqref{GENERAL CASE7} then this integral can be associated with term $(\widetilde{\rho}_0(n-1)|\hat{O})$. Moreover, we can see how quasicontinuous part $\{\omega_j=a+\frac{(j-1)(b-a)}{N-1}:j=1,...,N\}$ joins with continuous part $\omega\in[0,\infty)$ in the limit $N\rightarrow\infty$ expressed in integrals $\int_{[a,b]\cup [0,\infty)}$ and $\int\int_{[a,b]\times[a,b]\cup [0,\infty)\times[0,\infty)}$ of Eq. \eqref{GENERAL CASE7}.

Summing up, given a quantum system $S$ having both types of spectrum, discrete and continuous, with a constrictive and markovian Frobenius-Perron operator $P$ associated with the temporal evolution of its classical analogue we have

\begin{itemize}
\item $S$ is ergodic $\Longrightarrow$ nondiagonal term of the discrete part of any quantum mean value (i.e. the second term of the right hand of Eq. \eqref{GENERAL CASE1}) is a periodic function of $n$.
\item $S$ is mixing $\Longleftrightarrow$ discrete part is quasicontinuous.

\end{itemize}

\section{QSDT applications}
In this section we apply the QSDT to two examples\footnote{We omit exact case since by QSDT (Theorem 10) we have $\hat{U}$ exact $\Longleftrightarrow$ $\hat{U}$ mixing.} to illustrate its physical relevance: Microwave billiards and a phenomenological Gamow model. As we pointed out in section 4 and 5.1.1, both examples are open quantum systems that can be described by an effective non-Hermitian Hamiltonian.

\subsection{Quantum ergodic and quantum mixing: Microwave billiards}

Microwave billiards are special examples of scattering systems \cite{stockmann, ste95}. Typically, to determine the spectrum of such systems antennas are used as scattering channels. An external coupling determines the resonances positions and widths, and the spectrum of microwave billiards are modified by the presence of the coupling antennas. We can start with an expression for the scattering matrix (see \cite{stockmann} pag. 221)

\begin{equation}\label{MB1}
\hat{S}=\hat{1}-2i\hat{W}^{\dag}\frac{\hat{1}}{E-\hat{H}_0+i\hat{W}\hat{W}^{\dag}}\hat{W}
\end{equation}

\noindent where $\hat{1}$ is the identity matrix, $\hat{H}_0$ is the undisturbed Hamiltonian assumed to be a $N\times N$ truncated matrix and matrix elements $W_{nk}$ of $\hat{W}$ contain information on the coupling strengths of the $k$th channel to the $n$th resonance. The poles of scattering matrix are  eigenvalues of the effective Hamiltonian

\begin{equation}\label{MB2}
\hat{H}=\hat{H}_0-i\hat{W}\hat{W}^{\dag}
\end{equation}

\noindent Information on widths and shifts induced by antennas is completely contained in the eigenvalues of effective hamiltonian $\hat{H}$. Effective Hamiltonians of this type have been extensively used in nuclear physics \cite{Mah69, Lew91, Fyo97, Per96, Sok88, Sto98}. We are only interested in the limiting case of large coupling strengths where $\hat{H}$ is dominated by the term $-i\hat{W}\hat{W}^{\dag}$. The key is to select a basis where this term is diagonal and $\hat{H}_0$ is treated as a perturbation. If we have $K$ channels without loss of generality we may assume that the $K$ vectors $\hat{w}_k$ with components $W_{nk}$ are mutually orthogonal,

\begin{equation}\label{MB3}
\hat{w}_k^{\dag}\hat{w}_l=\sum_n W_{nk}^*W_{nl}=|\hat{w}_k|^{2}\delta_{kl}
\end{equation}

\noindent Then we can take the orthogonal basis formed by the $K$ vectors $\hat{\nu}_k=\frac{\hat{w}_k}{|\hat{w}_k|}$ and the $N-K$ vectors $\hat{u}_{\alpha}$ where $\hat{H}_0$ is diagonal in subspace generated by vectors $\hat{u}_{\alpha}$ (see \cite{stockmann} pag. 221 and 222). Using first order perturbation theory the eigenvalues of $\hat{H}$ are

\begin{equation}\label{MB4}
\begin{split}
&E_l=\hat{\nu}_l^{\dag}\hat{H}_0\hat{\nu}_l-i|\hat{w}_l| \,\,\,\,\,\,\,\,\,\,\,\,\,\,\,\,\,\,\,\,\, l\leq K\\
&E_l=\hat{u}_l^{\dag}\hat{H}_0\hat{u}_l \,\,\,\,\,\,\,\,\,\,\,\,\,\,\,\,\,\,\,\,\,\,\,\,\,\,\,\,\,\,\,\,\,\,\,\,\,\,\,\,\, l> K
\end{split}
\end{equation}
\noindent Now we can consider the Hamiltonian

\begin{equation}\label{MB5}
\hat{\mathcal{H}}=\sum_{l=1}^{K}E_{l}|\nu_l\rangle
\langle\widetilde{\nu_l}|+\sum_{m=K+1}^{N}E_{m}|u_m\rangle
\langle\widetilde{u_m}|
\end{equation}

\noindent which is a first order approximation of effective Hamiltonian $\hat{H}$. We have renamed $\hat{\nu}_l$, $\hat{u}_m$ as $|\nu_l\rangle$, $|u_m\rangle$ respectively, and Hamiltonian $\hat{\mathcal{H}}$ has been expressed in two sums to analyze different cases according to the number of channels $K$. Also, left vectors $\langle\widetilde{\nu_l}|,\langle\widetilde{u_m}|$ are those defined by Eqns. \eqref{complex discrete1} and \eqref{complex discrete2}. We consider an initial state $\hat{\rho}$ given by
\begin{equation}\label{MB6}
\hat{\rho}=\sum_{l=1}^{K}\sum_{l^{\prime}=1}^{K}\rho_{ll^{\prime}}|\nu_l\rangle
\langle\widetilde{\nu_{l^{\prime}}}|+\sum_{m=K+1}^{N}\sum_{m^{\prime}=K+1}^{N}\rho_{mm^{\prime}}|u_m\rangle
\langle\widetilde{u_{m^{\prime}}}|+\left\{\sum_{\sigma=1}^{K}\sum_{\lambda=K+1}^{N}\rho_{\sigma\lambda}|\nu_{\sigma}\rangle
\langle\widetilde{u_{\lambda}}|+h.c.\right\}
\end{equation}
\noindent where $\rho_{ij}=\langle i|\hat{\rho}|j\rangle$ for all $i,j=1,...,N$ and first two sums of Eq. \eqref{MB6} are diagonal blocks corresponding to the subspaces spanned by $\{|\nu_l\rangle\}_{l=1}^K$ and $\{|u_m\rangle\}_{m=K+1}^N$ respectively. Third term of Eq. \eqref{MB6} contains nondiagonal elements of $\hat{\rho}$ which connect subspaces spanned by $\{|\nu_l\rangle\}_{l=1}^K$, $\{|u_m\rangle\}_{m=K+1}^N$ and $h.c.$ denotes the hermitian conjugate operation. We rename $\hat{\nu}_l^{\dag}\hat{H}_0\hat{\nu}_l$ and $\hat{u}_m^{\dag}\hat{H}_0\hat{u}_m$ as $\gamma_l$ and $\omega_m$ for all $l,m$. Again, operator $\hat{U}(n)$ is given by $\hat{U}(n)=e^{-i\frac{\hat{\mathcal{H}}}{\hbar}\alpha n}$ so $\hat{\rho}$ after $n$ steps is

\begin{equation}\label{MB7}
\begin{split}
&\hat{\rho}(n)=\hat{U}(n)\hat{\rho}\hat{U}(n)^{\dag}=e^{-i\frac{\hat{\mathcal{H}}}{\hbar}\alpha n}\hat{\rho}e^{i\frac{\hat{\mathcal{H}^{\dag}}}{\hbar}\alpha n}
=\sum_{l=1}^{K}\sum_{l^{\prime}=1}^{K}\rho_{ll^{\prime}}e^{-\frac{(|\hat{w}_l|+|\hat{w}_{l^{\prime}}|)\alpha n}{\hbar}}e^{\frac{-i(\gamma_l-\gamma_{l^{\prime}})\alpha n}{\hbar}}|\nu_l\rangle  \langle\widetilde{\nu_{l^{\prime}}}|\\
&+\sum_{m=K+1}^{N}\sum_{m^{\prime}=K+1}^{N}\rho_{mm^{\prime}}e^{\frac{-i(\omega_m-\omega_{m^{\prime}})\alpha n}{\hbar}}|u_m\rangle
\langle\widetilde{u_{m^{\prime}}}|+\left\{\sum_{\sigma=1}^{K}\sum_{\lambda=K+1}^{N}\rho_{\sigma\lambda}e^{-\frac{|\hat{w}_{\sigma}|\alpha n}{\hbar}}e^{\frac{-i(\gamma_{\sigma}-\omega_{\lambda})\alpha n}{\hbar}}|\nu_{\sigma}\rangle
\langle\widetilde{u_{\lambda}}|+h.c.\right\}
\end{split}
\end{equation}

\noindent Then the mean value of an observable $\hat{O}$ in $\hat{\rho}$ after $n$ steps is

\begin{equation}\label{MB8}
\begin{split}
&(\hat{\rho}(n)|\hat{O})=tr(\hat{\rho}(n)\hat{O})=\sum_{l=1}^{K}\sum_{l^{\prime}=1}^{K}\rho_{ll^{\prime}}O_{ll^{\prime}}
e^{-\frac{(|\hat{w}_l|+|\hat{w}_{l^{\prime}}|)\alpha n}{\hbar}}e^{\frac{-i(\gamma_l-\gamma_{l^{\prime}})\alpha n}{\hbar}}\\
&+\sum_{m=K+1}^{N}\sum_{m^{\prime}=K+1}^{N}\rho_{mm^{\prime}}O_{mm^{\prime}}
e^{\frac{-i(\omega_m-\omega_{m^{\prime}})\alpha n}{\hbar}}+\left\{\sum_{\sigma=1}^{K}\sum_{\lambda=K+1}^{N}\rho_{\sigma\lambda}O_{\sigma\lambda}
e^{-\frac{|\hat{w}_{\sigma}|\alpha n}{\hbar}}e^{\frac{-i(\gamma_{\sigma}-\omega_{\lambda})\alpha n}{\hbar}}+h.c.\right\}
\end{split}
\end{equation}
\noindent where $O_{ij}=\langle i|\hat{O}|j\rangle$ for all $i,j=1,...,N$. From Eq. \eqref{MB8} we see that first and third sums decay exponentially as $n\rightarrow\infty$ while second sum oscillates according to frequency differences $\omega_m-\omega_{m^{\prime}}$. This remark is crucial in order to analyze different cases according to number of channels $K$. By QSDT we have

\begin{equation}\label{MB9}
(\hat{\rho}(n)|\hat{O})=\sum_{i=1}^{r}(\hat{\rho}|\hat{O}_{\alpha^{-n}(i)})(\hat{\rho}_i|\hat{O})+(\widetilde{\rho}_0(n-1)|\hat{O})
\end{equation}
\noindent Considering the previous remark and comparing Eqns. \eqref{MB8} and \eqref{MB9} we have

\begin{equation}\label{MB10}
\begin{split}
&\sum_{i=1}^{r}(\hat{\rho}|\hat{O}_{\alpha^{-n}(i)})(\hat{\rho}_i|\hat{O})=
\sum_{m=K+1}^{N}\sum_{m^{\prime}=K+1}^{N}\rho_{mm^{\prime}}O_{mm^{\prime}}e^{\frac{-i(\omega_m-\omega_{m^{\prime}})\alpha n}{\hbar}}\\
&(\widetilde{\rho}_0(n-1)|\hat{O})=\sum_{l=1}^{K}\sum_{l^{\prime}=1}^{K}\rho_{ll^{\prime}}O_{ll^{\prime}}
e^{-\frac{(|\hat{w}_l|+|\hat{w}_{l^{\prime}}|)\alpha n}{\hbar}}e^{\frac{-i(\gamma_l-\gamma_{l^{\prime}})\alpha n}{\hbar}}
+\left\{\sum_{\sigma=1}^{K}\sum_{\lambda=K+1}^{N}\rho_{\sigma\lambda}O_{\sigma\lambda}
e^{-\frac{|\hat{w}_{\sigma}|\alpha n}{\hbar}}e^{\frac{-i(\gamma_{\sigma}-\omega_{\lambda})\alpha n}{\hbar}}+h.c.\right\}
\end{split}
\end{equation}

\noindent First line of Eq. \eqref{MB10} shows the oscillatory part of $(\hat{\rho}(n)|\hat{O})$ while second line of Eq. \eqref{MB10} expresses the decay terms of $(\hat{\rho}(n)|\hat{O})$. This can be considered as a ``global" QSDT characterization of first order Hamiltonian spectrum of microwave billiards.

Going into more detail, we consider a rectangle microwave billiard whose quantum mean values (at first order) are given by Eq. \eqref{MB10}. In this case unperturbed Hamiltonian $\hat{H}_0$ has frequencies $\frac{\gamma_l}{\hbar}=\frac{\hat{\nu}_l^{\dag}\hat{H}_0\hat{\nu}_l}{\hbar}$, $\frac{\omega_m}{\hbar}=\frac{\hat{u}_m^{\dag}\hat{H}_0\hat{u}_m}{\hbar}$ which can be considered rationally related\footnote{This assumption is reasonable since dimensions of billiard can be adjusted such that frequencies are rationally related. Indeed, since only frequencies $\frac{\omega_m}{\hbar}$ are related with term $\sum_{i=1}^{r}(\hat{\rho}|\hat{O}_{\alpha^{-n}(i)})(\hat{\rho}_i|\hat{O})$ (see first line of Eq. \eqref{MB10}) then it is enough to consider that frequencies $\frac{\omega_m}{\hbar}$ are rationally related.  } for all $l=1,...,K$; $m=K+1,...,N$ then by arguments of consideration (B) of section 5.1 we have double sum of the first line of Eq. \eqref{MB10} is a periodic function of $n$. Now varying number of channels $K$ we can obtain different chaotic transitions from integrable regime to chaotic one.

We begin with $K=0$ that corresponds to integrable case $\hat{W}=0$. Then $|\hat{w}_l|=0$ for all $l=1,...,K$ in Eq. $\eqref{MB4}$ and quantum mean values have no terms going to zero, e.g. all terms of $(\hat{\rho}(n)|\hat{O})$ in Eq. \eqref{MB8} are oscillatory. Since there is no term that tends to zero then QSDT does not apply in this case.

Case $1\leq K < N-1$. In this case $\hat{W}\neq0$ and we have an exponential decay of $(\widetilde{\rho}_0(n-1)|\hat{O})$ with characteristic times $\tau_{ll^{\prime}}=\frac{\hbar}{|\hat{w}_l|+|\hat{w}_{l^{\prime}}|}$, $\tau_{\sigma}=\frac{\hbar}{|\hat{w}_{\sigma}|}$ for all $l,l^{\prime},\sigma=1,...,K$. Since frequencies $\omega_l$ are rationally related for all $m=K+1,...,N$ then right hand of first line of Eq. \eqref{MB10} is a periodic function of $n$. Therefore, from QSDT we have that microwave billiard is \emph{ergodic} for $1\leq K < N-1$ (see end of section 5.1). This case corresponds to pseudointegrable regime ($K=1$) and chaotic regime ($K>1$).

Case $K=N-1$. When $K=N-1$ we have $(\widetilde{\rho}_0(n-1)|\hat{O})$ decays exponentially with characteristic times $\tau_{ll^{\prime}}=\frac{\hbar}{|\hat{w}_l|+|\hat{w}_{l^{\prime}}|}$, $\tau_{\sigma}=\frac{\hbar}{|\hat{w}_{\sigma}|}$ for all $l,l^{\prime},\sigma=1,...,N-1$ and right hand of the first line of Eq. \eqref{MB10} has only one term, $\rho_{NN}O_{NN}$. Then by QSDT it follows that $r=1$ and this means that microwave billiard is \emph{mixing}. Therefore, microwave billiard is \emph{mixing} for $K=N-1$. This case corresponds to a fully chaotic regime with maximum number of terms that decay exponentially.


The case $K=N$ is physically excluded because all quantum mean values $(\hat{\rho}(n)|\hat{O})$ can not tend to zero.

Therefore, application of QSDT to microwave billiards says that increasing of coupling channel number $K$ can be interpreted as chaotic transitions from integrable ($K=0$) to ergodic ($1<K\leq N-1$) and from ergodic to mixing ($K=N-1$).

\subsection{Quantum mixing case: A phenomenological Gamow model}
Phenomenological Gamow model type \cite{Omnes1, Omnes2} is perhaps one of simplest and more
illustrative examples of decoherence and approach to equilibrium of a quantum system. Our system is a
single quantum oscillator embed in a environment composed of a large bath of noninteracting quantum oscillators which can be considered a continuum. Degeneracies of this system render it convenient for application of Hamiltonian analytic extension (see \cite{Omnes1, Gadella, letterpolos, Ordonito-dec}) in order to obtain a non-Hermitian effective Hamiltonian \footnote{In fact, in this non-Hermitian case we have $\hat{H}=\sum_{k=0}^{\infty}z_{k}|k\rangle
\langle\widetilde{k}|\neq \sum_{m=0}^{\infty}z_{m}^{\ast}|m\rangle
\langle\widetilde{m}|=\hat{H}^{\dag}$ due presence of complex eigenvalues $z_k$ ($k\neq0$) with nonzero imaginary parts.} given by

\begin{equation}\label{QSDT mixing1}
\hat{H}=\sum_{k=0}^{\infty}z_{k}|k\rangle
\langle\widetilde{k}|
\end{equation}
where $z_{k}=k(\omega_0-i\gamma_0)$ are complex eigenvalues (except $z_0=\omega_0$), $k=0,1,2,...$ Natural frequency of single oscillator is $\omega_0$ and $\gamma_0$ is associated with relaxation time $t_R$ by $t_R=\frac{\hbar}{\gamma_0}$ (see \cite{Omnes2} pag. 288). In other words, $\gamma_0$ is inversely proportional to decay rate of dumping of single oscillator. We consider an initial state $\hat{\rho}$ given by
\begin{equation}\label{QSDT mixing2}
\hat{\rho}=\sum_{k=0}^{\infty}\sum_{m=0}^{\infty}\rho_{km}|k\rangle
\langle\widetilde{m}|
\end{equation}
where $\rho_{km}=\langle k|\widehat{\rho}|m\rangle$ and
$\rho_{kk}\geq0$\,,\,$\rho_{mk}=\rho_{km}^{\ast}$ with $k,m \in \mathbb{N}_0$.
State $\hat{\rho}$ after $n$ time steps is given by
\begin{equation}\label{QSDT mixing3}
\begin{split}
&\hat{\rho}(n)=\rho_{00}|0\rangle\langle 0|+\sum_{k=1}^{\infty}\rho_{kk}e^{-2k\frac{\gamma_0}{\hbar}n}|k\rangle\langle k|+\sum_{k=0}^{\infty}\sum_{m=0, k\neq m}^{\infty}\rho_{km} e^{-\frac{\gamma_0}{\hbar}(k+m)n}|k\rangle\langle m|
\end{split}
\end{equation}
Then mean value of an observable $\hat{O}$ in $\hat{\rho}$ after $n$ time steps is
\begin{equation}\label{QSDT mixing4}
\begin{split}
&(\hat{\rho}(n)|\hat{O})=tr(\hat{\rho}(n)\hat{O})=\sum_{k=0}^{\infty}\langle\widetilde{k}|\hat{\rho}(n)\hat{O}|k\rangle=
\sum_{k=0}^{\infty}\{\hat{\rho}(n)\hat{O}\}_{kk}=\sum_{k=0}^{\infty}\sum_{m=0}^{\infty}\{\hat{\rho}(n)\}_{km}\{\hat{O}\}_{mk}=\\
&=\rho_{00}O_{00}+\sum_{k=1}^{\infty}\rho_{kk}O_{kk}e^{-2k\frac{\gamma_0}{\hbar}n}+\sum_{k=0}^{\infty}\sum_{m=0, k\neq m}^{\infty}\rho_{km}O_{km} e^{-\frac{\gamma_0}{\hbar}(k+m)n}
\end{split}
\end{equation}
where $\langle n|\hat{O}|m\rangle=O_{nm}$ with $k,m \in \mathbb{N}_0$ and parenthesis with subindexes  $\{\hat{\rho}(n)\hat{O}\}_{kk},\{\hat{\rho}(n)\}_{km},\{\hat{O}\}_{mk}$ are notations for corresponding matrix elements of $\hat{\rho}(n)\hat{O}, \hat{\rho}(n)$ and $\hat{O}$ respectively. On the other hand by, QSDT we have

\begin{equation}\label{QSDT mixing5}
\begin{split}
&\sum_{i=1}^{r}(\hat{\rho}_{\psi}|\hat{O}_{\alpha^{-n}(i)})(\hat{\rho}_i|\hat{O}) + (\widetilde{\rho}_0(n-1)|\hat{O})=\\
&=\rho_{00}O_{00}+\sum_{k=1}^{\infty}\rho_{kk}O_{kk}e^{-2k\frac{\gamma_0}{\hbar}n}+\sum_{k=0}^{\infty}\sum_{m=0, k\neq m}^{\infty}\rho_{km}O_{km} e^{-\frac{\gamma_0}{\hbar}(k+m)n}
\end{split}
\end{equation}

\nd Since first term $\rho_{00}O_{00}$ on right hand of Eq. \eqref{QSDT mixing5} is constant and remaining terms tends to zero when $n\rightarrow \infty$ (due presence of decreasing exponentials) then sum on left hand of Eq. \eqref{QSDT mixing5} must consists of only one term, i.e. $r=1$. From this fact it follows that $\widetilde{\rho}_0(n-1)=\hat{\rho}(n)-\rho_{00}|0\rangle\langle 0|$ is associated with decay modes $z_k$ ($k\neq0$).
Therefore, from Theorem 10 (II, III) it follows that Gamow model is \emph{mixing}. Moreover, we can obtain weak
limit of state $\hat{\rho}$ (see \cite{0} section 6.3 def. B). From Eq. \eqref{QSDT mixing5} we have

\begin{equation}\label{QSDT mixing6}
\begin{split}
&(\hat{\rho}(n)|\hat{O})=\sum_{i=1}^{r}(\hat{\rho}_{\psi}|\hat{O}_{\alpha^{-n}(i)})(\hat{\rho}_i|\hat{O}) + (\widetilde{\rho}_0(n-1)|\hat{O})
\longrightarrow\rho_{00}O_{00}=\rho_{00}(\hat{\rho}_{|0\rangle\langle 0|}|\hat{O}) \,\,\,\ for \,\,\,\ n\longrightarrow \infty
\end{split}
\end{equation}

\nd where $\hat{\rho}_{|0\rangle\langle 0|}=|0\rangle\langle 0|$. That is, $\hat{\rho}_{w}=\rho_{00}|0\rangle\langle 0|$ is the weak limit of $\hat{\rho}$.\footnote{Non-normalization of $\hat{\rho}_{w}=\rho_{00}|0\rangle\langle 0|$ is a consequence of non-Hermiticity of Hamiltonian $\hat{H}$ given by Eq. \eqref{QSDT mixing1}. This is so because if $\hat{O}=\hat{I}$ from Eq. \eqref{QSDT mixing4} we have
$tr(\rho(n))$ is decreasing and $tr(\rho(n))=(\rho(n)|\hat{I})\rightarrow\rho_{00}$ when $n\rightarrow\infty$. Therefore, $0<\rho_{00}\leq1$ and $\rho_{00}=1$ $\Longleftrightarrow$ $\rho_{kk}=0$ for all $k\neq0$.}

\section{Conclusions}
Assuming that the classical limit of a quantum system
has a constrictive Markovian Frobenius-Perron operator associated with its classical evolution $T$ we presented a quantum version of Spectral Decomposition Theorem of Dynamical
Systems (Theorems 6 and 7) we called Quantum Spectral Decomposition Theorem (QSDT, Theorem 9). QSDT gives a
representation of expectation values of all observable $\hat{O}$ characterizing ergodic level of QEH by presence of an oscillatory term (see Eq. \eqref{QSDT consequence ergodic3}) and at the same time contains Riemann-Lebesgue Lemma for van Hove
observables in the mixing case (see Eq.
\eqref{CONTINUOUS SPECTRA7}). Moreover, in the mixing case QSDT
provides a physical interpretation of ``homogenization"
(see Eqns. \eqref{QSDT consequence mixing} and \eqref{QSDT
consequence mixing2}) which is represented by a pure state $\hat{\rho}_1$ that is weak limit of $\hat{\rho}$.

Considering that quantum evolution is given by $\hat{U}(n)=e^{-i\frac{\hat{H}}{\hbar}\alpha n}$, i.e. a discretized evolution at constants intervals where the parameter $\alpha\in\mathbb{R}$ defines the time step, QSDT links QEH levels with spectrum. More precisely, when the spectrum is discrete we have that quasi-periodicity of quantum mean values is a necessary condition for ergodicity and linear dependence in the ring of integers of energy frequencies $\omega_1=\frac{E_1}{\hbar}, \omega_2=\frac{E_2}{\hbar},...,\omega_N=\frac{E_N}{\hbar}$ is a sufficient condition for non-ergodicity (see condition $(\star\star)$ of section 5.1).

For the complex discrete case QSDT gives us a characterization of quasi-stationary states (see Eq. \eqref{complex discrete7}) where exponentially decreasing terms of any mean value $(\hat{\rho}(n)|\hat{O})$ are associated with decay term $(\widetilde{\rho}_0(n-1)|\hat{O})$. This can be considered as a QSDT characterization of open quantum systems described by an effective and non-Hermitian Hamiltonian.

For the continuous case we have that mean values of \emph{van Hove observables} are composed by two terms, diagonal and non-diagonal, whose physical interpretation can be analyzed by means of QSDT decomposition. The terms $\int_{0}^{\infty}\rho(\omega)^{\ast}O(\omega)d\omega$ and $\int_{0}^{\infty}\int_{0}^{\infty}\rho(\omega,\omega^{\prime})^{\ast}O(\omega,\omega^{\prime})e^{-i\frac{(\omega-\omega^{\prime})}{\hbar}n}d\omega d\omega^{\prime}$ can be identified with $(\hat{\rho}|\hat{O}_1)(\hat{\rho}_1|\hat{O})$ and $(\widetilde{\rho}_0(n-1)|\hat{O})$ respectively, see Eq. \eqref{CONTINUOUS SPECTRA7}. In other words, mixing case corresponds to continuous spectrum with a van Hove algebra.
When both spectra are present QSDT give us two conditions (see section 5.3).
\begin{itemize}
\item Quantum system is ergodic then non-diagonal term of discrete part of spectrum is a periodic function of $n$, see Eq. \eqref{GENERAL CASE4}.
\item Quantum system is mixing if and only if discrete part of spectrum can be approximated by a quasicontinuous, see Eqns. \eqref{GENERAL CASE5} and \eqref{GENERAL CASE6}.
\end{itemize}
As we pointed out in section 5.3, the general case is not the superposition of discrete and continuous cases simultaneously. In the ergodic case periodic term is provided by discrete part of spectrum while decay term is contributed by continuous part of spectrum. On the other hand, in the mixing case both continuous part and discrete part provide decay terms associated with $(\widetilde{\rho}_0(n-1)|\hat{O})$ where discrete part is necessarily quasi-continuous. Roughly speaking, we can say that continuous part ``forces" to discrete part to be periodic in the ergodic case and to be quasi-continuous in the mixing case.

In Section 6 we apply QSDT to two examples, microwave billiards and a phenomenological Gamow type model. In the first case QSDT allows us to characterize chaotic transitions, from integrable regime to chaotic one, of microwave billiards in terms of the number of channels $K$. When there is no channels ($K=0$) mean values $(\widehat{\rho}(n)|\widehat{O})$ are oscillatory which corresponds to integrable regime. For a number of channels $K$ such that $1\leq K< N-1$ \footnote{$N$ is the dimension of truncated matrix of undisturbed Hamiltonian $\widehat{H}_0.$} QSDT says that microwave billiard is ergodic which corresponds to pseudointegrable regime ($K=1$) and chaotic regime ($K>1$). This results due presence of terms that decay exponentially with a rate inversely proportional to the coupling strength $|\hat{w}_{\sigma}|$ for all $\sigma=1,...,K$ (see Eq. \eqref{MB10}). This seems physically reasonable since a large coupling strength $|\hat{w}_{\sigma}|\gg1$ implies a high scattering channel which corresponds to a small characteristic time $\tau_{\sigma}=\frac{\hbar}{|\hat{w}_{\sigma}|}\ll1$, i.e. the exponential decay is fast. When $K=N-1$ QSDT says that microwave billiard is mixing which corresponds to fully chaotic regime where mean values $(\widehat{\rho}(n)|\widehat{O})$ contain a maximum number of terms that decay exponentially (see Eq. \eqref{MB10}). We summarize QSDT characterization of  microwave billiards in the following table.
\vskip1truecm
\centerline{TABLE I: QSDT characterization of microwave billiards\footnote{Again, we consider dimensions of billiard can be adjusted such that unperturbed Hamiltonian $\hat{H}_0$ has frequencies $\frac{\gamma_l}{\hbar}=\frac{\hat{\nu}_l^{\dag}\hat{H}_0\hat{\nu}_l}{\hbar}$, $\frac{\omega_m}{\hbar}=\frac{\hat{u}_m^{\dag}\hat{H}_0\hat{u}_m}{\hbar}$ which are rationally related.}}\vskip0.5truecm

\begin{tabular}{||l | c | r||}
\hline
\hline
&              &            \\
Number of channels $K$ & Quantum mean values $(\hat{\rho}(n)|\hat{O})$  & \,\,\,\,\,\, QEH level, regime \,\,\,\,\,\,\,\,\,\,\,\, \\
&              &            \\
\hline
&              &            \\
\,\,\,\,\,\,\,\,\,\,\,\,\,\,\,\, $K=0$  &     oscillatory          &     none, integrable  \,\,\,\,\,\,\,\,\,\,\,\,     \\
&              &            \\
\hline
&              &            \\
\,\,\,\,\,\,\,\, $1\leq K< N-1$  &     exponential decay terms         &    $\mathbf{ergodic}$, chaotic   \,\,\,\,\,\,\,\,\,\,\,\,    \\
&              &            \\
\hline
&              &            \\
  &     maximum number of           &    $\mathbf{mixing}$, \,\,\,\,\,\,\,\,\,\,\,\,\,\,\,\,\,\,\,\,\,\,\,\,\,\,     \\
\,\,\,\,\,\,\,\,\,\,\,\, $K=N-1$ &    exponential terms          &   fully chaotic  \,\,\,\,\,\,\,\,\,\,\,\,\,\,\,\,\,\,        \\
&              &            \\
\hline
\hline
\end{tabular}

\vspace{0.5cm}

\noindent From Table I we see that increasing of number of channels implies an increasing of chaotic level. System starts with an integrable regime for $K=0$, enters to a chaotic ergodic regime when $1\leq K< N-1$ and finally reaches a fully chaotic mixing regime for $K=N-1$, i.e. as number of antennas increases system becomes more chaotic. This can be considered as a QSDT characterization of microwave billiards in the limit of large coupling strengths.

In the case of Gamow model QSDT decomposition determines its mixing level, see Eq. \eqref{QSDT mixing5}. Moreover, given an initial state $\hat{\rho}$ its weak limit $\hat{\rho}_w=\rho_{00}|0\rangle\langle 0|$ can not be normalized due non-Hermiticity of effective Hamiltonian $\hat{H}$, see Eqns. \eqref{QSDT mixing1} and \eqref{QSDT mixing6}.








We hope that all these features provided by QSDT can be be useful in shedding light on quantum chaos in future studies through more examples and
theoretical essays.

\end{document}